\documentclass{article}
\usepackage{PRIMEarxiv}
\usepackage[numbers,sort&compress]{natbib}
\PassOptionsToPackage{hypertexnames=false}{hyperref}
\usepackage[T1]{fontenc}
\usepackage[utf8]{inputenc}
\usepackage{amsmath,amssymb,mathtools,bm}
\usepackage{graphicx,booktabs,array,longtable,calc}
\usepackage{microtype}
\usepackage{xcolor}
\usepackage{capt-of}
\usepackage{placeins}
\usepackage{setspace}
\usepackage[font={footnotesize,stretch=1.0},labelfont=bf]{caption}
\usepackage{xurl}
\usepackage[hidelinks]{hyperref}
\usepackage[nameinlink,noabbrev]{cleveref}
\graphicspath{{figures/}}

\usepackage{tikz}

\usepackage{orcidlink}
\usepackage{xstring}
\usepackage{amsthm}

\newcommand{\ArxivNativeLabel}{}
\newcommand{\ArxivRefAdapter}[1]{%
  \IfBeginWith{#1}{SI-}{%
    \StrGobbleLeft{#1}{3}[\ArxivNativeLabel]%
    \ArxivOriginalRef{\ArxivNativeLabel}%
  }{%
    \ArxivOriginalRef{#1}%
  }%
}
\newcommand{\ArxivEqrefAdapter}[1]{%
  \IfBeginWith{#1}{SI-}{%
    \StrGobbleLeft{#1}{3}[\ArxivNativeLabel]%
    \ArxivOriginalEqref{\ArxivNativeLabel}%
  }{%
    \ArxivOriginalEqref{#1}%
  }%
}
\AtBeginDocument{%
  \let\ArxivOriginalRef\ref
  \let\ArxivOriginalEqref\eqref
  \let\ref\ArxivRefAdapter
  \let\eqref\ArxivEqrefAdapter
}

\newcommand{\BeginArxivSupplement}{%
  \let\ref\ArxivOriginalRef
  \let\eqref\ArxivOriginalEqref
  \clearpage
  \setcounter{section}{0}%
  \setcounter{subsection}{0}%
  \setcounter{equation}{0}%
  \setcounter{figure}{0}%
  \setcounter{table}{0}%
  \setcounter{page}{1}%
  \renewcommand{\thepage}{S\arabic{page}}%
  \renewcommand{\thesection}{S\arabic{section}}%
  \numberwithin{equation}{section}%
  \renewcommand{\theequation}{S\arabic{section}.\arabic{equation}}%
  \renewcommand{\thefigure}{S\arabic{figure}}%
  \renewcommand{\thetable}{S\arabic{table}}%
  \providecommand{\theHsection}{\arabic{section}}%
  \providecommand{\theHsubsection}{\arabic{section}.\arabic{subsection}}%
  \providecommand{\theHequation}{\arabic{equation}}%
  \providecommand{\theHfigure}{\arabic{figure}}%
  \providecommand{\theHtable}{\arabic{table}}%
  \renewcommand{\theHsection}{SI.\arabic{section}}%
  \renewcommand{\theHsubsection}{SI.\arabic{section}.\arabic{subsection}}%
  \renewcommand{\theHequation}{SI.\arabic{section}.\arabic{equation}}%
  \renewcommand{\theHfigure}{SI.\arabic{figure}}%
  \renewcommand{\theHtable}{SI.\arabic{table}}%
  \fancyhf{}%
  \fancyhead[L]{Supplementary Information}%
  \fancyhead[R]{B. Ni}%
  \cfoot{\thepage}%
  \renewcommand{\headrulewidth}{0pt}%
}

\newcommand{\ArxivSupplementTitle}{%
  \phantomsection
  \pdfbookmark[1]{Supplementary Information}{supplementary-information}%
  \thispagestyle{fancy}%
  \par
  \hrule height 2pt
  \vspace{0.25in}
  \begin{center}
    {\normalfont\LARGE Supplementary Information:\\
    Quantitative Eshelby stability: local and finite-shape theory through AI-assisted mechanics\par}
  \end{center}
  \vspace{0.29in}
  \hrule height 2pt
  \vspace{0.22in}
  \begin{center}
    {\large \orcidlink{0000-0003-2537-591X}\,\textbf{Bo Ni}$^{1,2,*}$\par}
    \vspace{0.55em}
    {\normalsize
    $^{1}$Department of Materials Science and Engineering, Carnegie Mellon University,\\
    5000 Forbes Ave., Pittsburgh, PA 15213, USA\\
    $^{2}$Generative Mechanics Laboratory (GenMech Lab), independent research initiative\\[0.45em]
    $^{*}$Correspondence: \href{mailto:nibo@cmu.edu}{nibo@cmu.edu}\par}
  \end{center}
  \vspace{0.35in}
}

\newcommand{\EndArxivSupplement}{}

\newtheorem{theorem}{Theorem}[section]
\newtheorem{lemma}[theorem]{Lemma}
\newtheorem{proposition}[theorem]{Proposition}

\title{\normalfont Quantitative Eshelby stability: local and finite-shape theory through AI-assisted mechanics}

\author{
  \orcidlink{0000-0003-2537-591X}\,\textbf{Bo Ni}$^{1,2,*}$ \\
  \\
  $^{1}$Department of Materials Science and Engineering, Carnegie Mellon University,\\
  5000 Forbes Ave., Pittsburgh, PA 15213, USA \\
  $^{2}$Generative Mechanics Laboratory (GenMech Lab), independent research initiative \\
  \\
  $^{*}$Correspondence: \href{mailto:nibo@cmu.edu}{nibo@cmu.edu}
}
\date{}

\hypersetup{
  pdftitle={Quantitative Eshelby stability: local and finite-shape theory through AI-assisted mechanics},
  pdfauthor={Bo Ni}
}

\renewcommand{\headrulewidth}{0pt}

\begin{document}
\maketitle

\begin{abstract}
How can human researchers use artificial intelligence (AI) to make progress on difficult mechanics problems while retaining control of physical significance, mathematical rigor, and validation? We explore this opportunity through a quantitative stability problem rooted in the Eshelby conjecture. For a uniform dilatational eigenstrain in a homogeneous isotropic solid, we establish a nonlinear two-sided relation between whole-interior strain heterogeneity and a geometry-only distance from the ellipsoid family, with half-order Sobolev scaling that captures morphological length scale. We then develop a complementary finite-shape theory: ellipsoidal subtraction yields an exact residual-field budget from which systematically convergent upper and lower bounds are obtained for a prescribed geometry. Applied to protruded, indented, and necked inclusions, the bounds agree with independent displacement-based finite-element calculations. The study also documents a human-led, AI-assisted research process in which researchers select consequential problems, shape targets, audit arguments, require independent checks, and interpret mechanics, while AI broadens analytical exploration and supports computation. This experience also led us to develop \emph{Open Problems in Mechanics}, a versioned community-facing map that makes selected mechanics questions, their known/open boundaries, and criteria for progress explicit.

\end{abstract}


\section{Introduction}
\label{sec:introduction}
Recent advances in artificial intelligence (AI) for scientific discovery, including mechanics, formal mathematical reasoning, and large-scale mathematical exploration, are expanding its role in research \cite{Wang2023AIDiscovery,NiBuehler2024,NiKaplanBuehler2024,Hubert2026,Georgiev2025}. For mechanics, this raises a timely question: how can researchers use increasingly capable AI to make progress on difficult theoretical problems while retaining control of physical significance, mathematical rigor, validation, and scientific understanding? Mechanics is a demanding setting because important problems often span general theory, quantitative computation, physical interpretation, and engineering use. Progress at one layer does not replace the others; increased AI capability likewise does not by itself ensure scientific understanding or verification \cite{MesseriCrockett2024,Hubert2026}.

The Eshelby inclusion problem provides a particularly suitable test. An ellipsoidal region carrying a uniform eigenstrain in an infinite homogeneous elastic solid has a spatially uniform constrained strain (Fig.~\ref{fig:overview}A) \cite{Eshelby1957}, a result foundational to micromechanics \cite{Parnell2016,Zhou2013Review}. It also motivated the Eshelby conjecture: to what extent does a uniform internal field characterize the ellipsoid? Rigidity results have resolved the weak conjecture in important settings and substantial parts of the strong conjecture in three dimensions \cite{KangMilton2008,Liu2008,Ammari2010StrongEshelby,YuanHuangWang2023}. Perturbation analyses have separately characterized infinitesimal uniformity-preserving directions and leading-order elastostatic changes under small interface perturbations \cite{Markenscoff1998a,Markenscoff1998b,Lagha2017}. The rigidity question is exact: when the field is uniform, what shapes are possible?

Real microstructures pose a quantitative counterpart. Precipitates and transforming regions can develop faceting, irregularity, coalescence, and interconnected morphologies as their size, coherency, interfacial energy, and elastic interactions evolve \cite{Marquis2001,Marquis2005,Karnesky2007,Lass2015}. Best-fit ellipsoids are useful in practice and can reproduce volume-averaged elastic quantities surprisingly well even for irregular inclusions, but the underlying nonellipsoidal geometry produces a heterogeneous internal field \cite{Karnesky2007,Zhong2021}. If such a morphology is represented by an ellipsoid, how much internal strain heterogeneity is neglected? Conversely, if the field is nearly uniform, how strongly is the nonellipsoidal shape constrained? The problem is not captured by a single geometric tolerance. Changes within the ellipsoid family preserve uniformity, equal-height features at different wavelengths need not have equal effects, and infinitesimal sensitivity alone does not control a finite morphology.

Nonellipsoidal inclusion mechanics has consequently been approached through exact or semi-analytical solutions for selected geometries and computational or averaged descriptions for more general shapes \cite{Rodin1996,Zou2010NonElliptical,Wu2023AverageEshelby}. These difficulties also expose the complementary roles of theory and computation. A theorem can control all admissible modes and reveal which geometric information is fundamentally required, something finite numerical sampling cannot establish. Shape-specific analysis can then turn that structure into quantitative predictions, while an independent elasticity calculation tests them through a different formulation. Engineering interpretation determines whether the resulting strain or stress variation matters for the approximation being used.

Here we develop these layers in one mechanics problem and use the research trajectory as a practical exploration of human-led, AI-assisted theoretical mechanics, not as a final prescription for scientific research. Human direction selected the problem, assumptions, targets, validation standards, and engineering interpretation; AI assistance broadened analytical search, computation, and checking (Fig.~\ref{fig:overview}B). We first establish a nonlinear two-sided stability relation near the ellipsoid family and identify the half-order boundary scale required for all-mode control (Fig.~\ref{fig:overview}C). We then derive a complementary finite-shape stability theory with systematically convergent bounds and test it on protruded, indented, and necked inclusions using independent finite-element calculations (Fig.~\ref{fig:overview}D). The Discussion considers what this process suggests for human-led, AI-assisted theoretical mechanics and introduces \emph{Open Problems in Mechanics} (Fig.~\ref{fig:overview}E), a community-facing map of selected theory-rich mechanics questions.

\begin{figure}[!htbp]
\centering
\includegraphics[width=\linewidth,keepaspectratio]{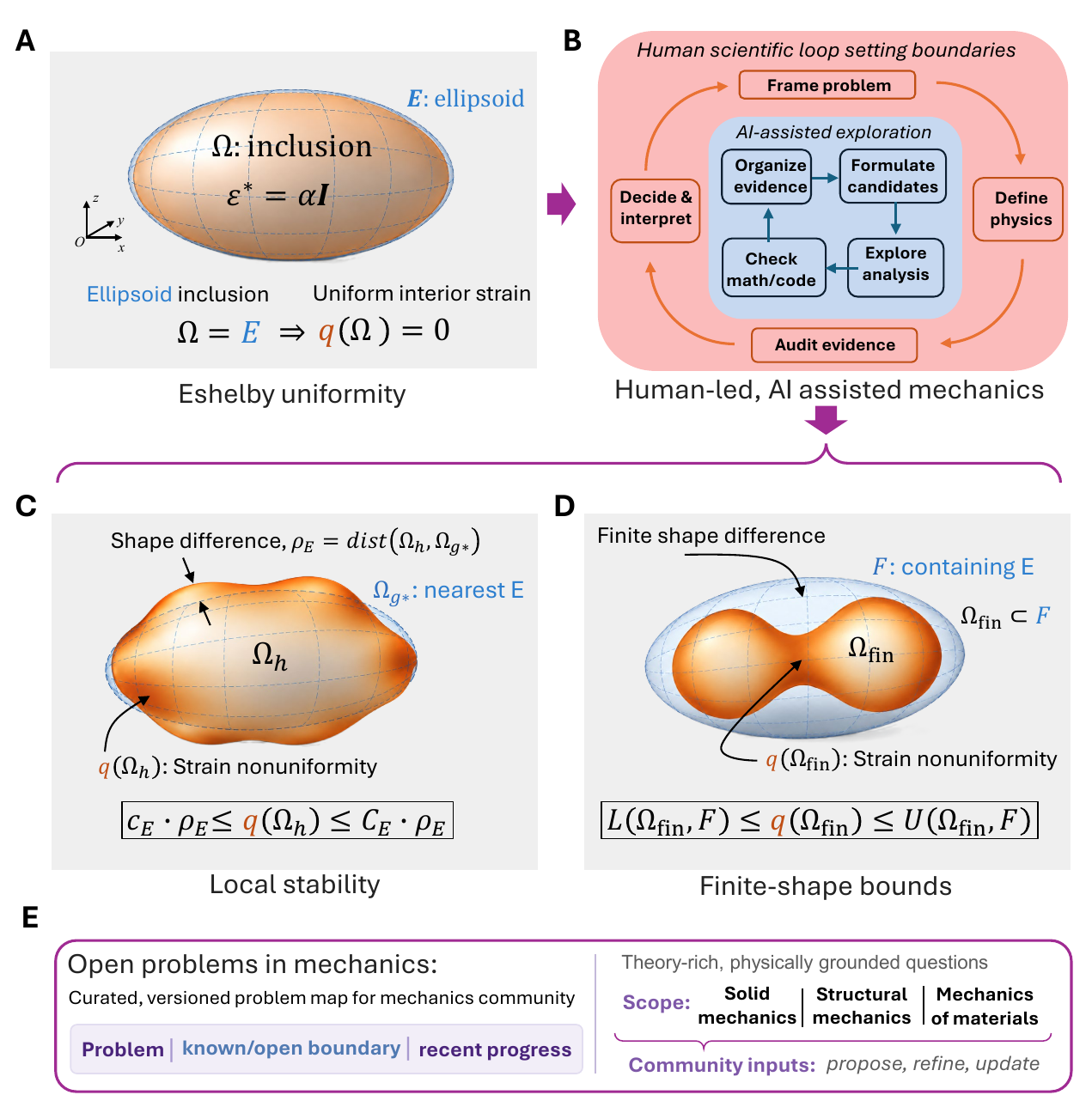}
\caption[From exact Eshelby uniformity to quantitative shape stability.]{\input{sections/caption_1}}
\label{fig:overview}
\end{figure}

\FloatBarrier

\section{Results}
\label{sec:results}
\subsection{Quantitative stability near ellipsoids}
\label{sec:quantitative-stability}
We consider a bounded inclusion \(\Omega\subset\mathbb R^3\) carrying a uniform dilatational eigenstrain \(\boldsymbol{\varepsilon}^{*}=\alpha\mathbf I\) in an infinite homogeneous isotropic solid. Inclusion and matrix have identical elastic moduli, the interface is perfectly bonded, and no remote load is applied. In this setting,

\begin{equation}
N_\Omega(\mathbf x)=\frac{1}{4\pi}\int_\Omega
\frac{d\mathbf y}{|\mathbf x-\mathbf y|},
\qquad
\boldsymbol{\varepsilon}
=-\alpha\beta\nabla^2N_\Omega ,
\tag{1}
\end{equation}

with

\begin{equation}
\beta=\frac{3K}{K+4\mu/3},
\tag{2}
\end{equation}

where \(K\) and \(\mu\) are the bulk and shear moduli. The Hessian \(\nabla^2N_\Omega\), the matrix of second spatial derivatives of the Newtonian potential \cite{DiFratta2016}, contains the complete shape dependence of the strain field.

We measure loss of uniformity by the normalized root-mean-square (RMS) strain fluctuation

\begin{equation}
q(\Omega)=
\frac{\|\boldsymbol{\varepsilon}-\langle\boldsymbol{\varepsilon}\rangle_\Omega\|_{L^2(\Omega)}}
{|\alpha|\beta\sqrt{|\Omega|}}
=
\frac{1}{\sqrt{|\Omega|}}
\min_{\mathbf M\in{\rm Sym}(3)}
\|\nabla^2N_\Omega-\mathbf M\|_{L^2(\Omega)} .
\tag{3}
\end{equation}

Here the minimization subtracts the best uniform symmetric tensor field before measuring the remaining volume-integrated variation. Because \(\Delta N_\Omega=-1\) inside the inclusion, the volumetric strain is constant for every admissible shape in this model, so \(q\) measures morphology-induced deviatoric heterogeneity.

Because every ellipsoid produces a uniform interior strain, the relevant geometric error is the departure from the ellipsoid family. To study this departure locally, fix a nondegenerate reference ellipsoid \(E\) and normalize lengths by its equivalent-volume radius. A nearby inclusion \(\Omega=\Omega_h\) is represented by displacing the boundary of \(E\) along its outward unit normal \(\mathbf n_E\),

\begin{equation}
\partial\Omega_h=
\{\mathbf x+h(\mathbf x)\mathbf n_E(\mathbf x):
\mathbf x\in\partial E\},
\tag{4}
\end{equation}

where \(h\) is the signed normal displacement. The reference \(E\) specifies the neighborhood of shapes under consideration. Within this neighborhood, we allow the approximating ellipsoid to vary: let \(\mathcal M_E\) contain the normal-height functions \(g\) whose graphs describe nearby exact ellipsoids, and define the half-order shape distance

\begin{equation}
\rho_E(h)=
\min_{g\in\mathcal M_E}\|h-g\|_{X_E},
\qquad
X_E\simeq H^{1/2}(\partial E).
\tag{5}
\end{equation}

Here \(H^{1/2}(\partial E)\) is the fractional Sobolev space of order one-half, and \(\simeq\) denotes equivalence of norms. The chosen norm \(\|\cdot\|_{X_E}\) measures both boundary amplitude and morphological length scale, giving greater weight to finer variations of the same amplitude. Its precise definition is given in \textbf{SI} Section~\ref{SI-pub:norms}, and Section 2.2 explains its mechanical significance. Minimizing over \(g\) removes changes in position, orientation, size, and aspect ratio that preserve an ellipsoidal shape, so \(\rho_E\) measures the remaining nonellipsoidal departure.

To define the local admissible neighborhood, we fix \(Z_E\) as a \(C^{1,1/2}\)-type norm on the normal height, controlling height, slope, and the half-H\"older variation of slope. Theorem~\ref{SI-pub:thm-local} in \textbf{SI} Section~\ref{SI-pub:local} is proved for any fixed \(C^{1,\gamma}\) class with \(0<\gamma<1\); \(\gamma=1/2\) is used here as a concrete convention. This strong norm controls geometric admissibility.

For each fixed reference ellipsoid \(E\), there exist an admissible neighborhood size \(\delta_0>0\) and positive finite constants \(c_E,C_E,A_E\) such that every perturbation in this neighborhood satisfies

\begin{equation}
\boxed{
(c_E-A_E\delta)\rho_E(h)
\le q(\Omega_h)
\le
(C_E+A_E\delta)\rho_E(h),
\qquad
\|h\|_{Z_E}\le\delta\le\delta_0 .
}
\tag{6}
\end{equation}

Here \(c_E\) and \(C_E\) are the optimal lower and upper response factors of the linearized problem after ellipsoidal changes are removed, while \(A_E\delta\) bounds the nonlinear correction. The neighborhood size \(\delta_0\) is chosen so that the lower coefficient remains positive.

Equation (6) controls the entire admissible shape class, including asymmetric features, localized imperfections, and combinations of variations across spatial scales (Fig.~\ref{fig:overview}C). For fixed \(E\) and norm conventions, the same constants apply throughout this class, independently of perturbation frequency, eigenstrain amplitude, and the elastic moduli after normalization by Eq. (3). The result therefore relates nonellipsoidal geometry to strain heterogeneity without prescribing a particular perturbation profile.

This class-wide control provides a basis for engineering tolerances. For a prescribed normalized RMS fluctuation allowance \(q_{\mathrm{allow}}\), every admissible shape satisfying \(\rho_E(h)\le q_{\mathrm{allow}}/(C_E+A_E\delta)\) meets that allowance. The lower bound provides complementary control: a prescribed minimum departure from the ellipsoid family guarantees a minimum strain heterogeneity throughout the admissible class. Conversely, achieving nearly uniform strain requires geometric proximity to that family. These implications support geometric specifications for limiting or generating strain heterogeneity across a broad range of admissible morphologies.

Section 2.2 gives explicit linearized constants for the sphere; for a general ellipsoid, their operator characterization is given in \textbf{SI} Section~\ref{SI-pub:theorem-statement}. A certified numerical tolerance at finite amplitude requires evaluated linearized bounds together with rigorous quantitative estimates of \(A_E\) and an admissible \(\delta_0\). These nonlinear quantities remain to be evaluated numerically.

The proof first identifies the boundary variations that produce zero linearized strain heterogeneity. These form the nine-dimensional tangent space of nearby ellipsoids: three translations of the center and six independent changes in the symmetric quadratic form defining the ellipsoid. The latter correspond to the six independent coefficients of a symmetric \(3\times3\) matrix, which determines the ellipsoid's axis lengths and orientation. This characterization agrees with earlier results showing that infinitesimal perturbations preserving uniform interior stress remain within the ellipsoid family \cite{Markenscoff1998a,Markenscoff1998b}. Linear control of the remaining directions supplies \(c_E\) and \(C_E\).

The nonlinear step fits an actual nearby ellipsoid and linearizes the field associated with the remaining nonellipsoidal deformation. After mapping the fields to the reference domain \(E\) and removing their uniform components, the normalized field remainder satisfies

\begin{equation}
\text{nonlinear error}
\le
C_{\mathrm{rem},E}\delta\rho_E(h).
\tag{7}
\end{equation}

Here the error is measured in the volume \(L^2\) norm on \(E\), and \(C_{\mathrm{rem},E}\) depends on the reference geometry and norm conventions. The factor \(\delta\) measures the overall perturbation size, while \(\rho_E(h)\) measures the residual departure from the ellipsoid family. This proportional dependence makes the remainder vanish for every nearby exact ellipsoid and controls it relative to the nonellipsoidal departure, completing the two-sided estimate. The complete formulation and proof are given in \textbf{SI} Sections~\ref{SI-pub:setting}--\ref{SI-pub:local}, with the relative remainder in Eq.~\eqref{SI-pub:eq-relative-remainder}; Methods Section~\ref{sec:method-local} summarizes the mechanical formulation and proof architecture.

\subsection{Morphological scale and sphere spectrum}
\label{sec:morphological-scale}
To interpret the half-order distance in Eq. (6), we first examine how morphological wavelength affects the strain response, then obtain explicit linearized bounds for the sphere. A broad protrusion and a fine corrugation can have the same height but generate different strain fields. Consider a locally flat reference interface \(z=0\), perturbed by

\begin{equation}
h(x)=a\cos(\kappa x),
\tag{8}
\end{equation}

where \(a\) is the amplitude and \(\kappa=2\pi/\lambda\) the wavenumber. Let \(\boldsymbol{\varepsilon}^{(1)}\) denote the strain perturbation linear in \(a\) induced by this boundary displacement. Figure~\ref{fig:scale-spectrum}A shows its magnitude \(|\boldsymbol{\varepsilon}^{(1)}|_F\) on the reference half-space \(z<0\). The subscript \(F\) denotes the Frobenius norm, the square root of the sum of squared tensor components. The strain perturbation decays into the inclusion as \(e^{\kappa z}\): its amplitude scales as \(|\alpha|\beta a\kappa\), while its penetration depth scales as \(\kappa^{-1}\). Direct integration gives

\begin{equation}
\int_{-\infty}^{0}
|\boldsymbol{\varepsilon}^{(1)}|_F^2\,dz
=
\frac{\alpha^2\beta^2a^2\kappa}{4}.
\tag{9}
\end{equation}

Thus the squared bulk response scales as amplitude squared times one power of wavenumber. A height RMS weights all wavelengths equally, while a slope RMS weights them by \(\kappa^2\). The intermediate high-frequency scaling is the homogeneous half-order Sobolev seminorm,

\begin{equation}
|h|_{\dot H^{1/2}}^2
\sim
\int |\boldsymbol{\xi}|\,
|\widehat h(\boldsymbol{\xi})|^2\,d\boldsymbol{\xi},
\tag{10}
\end{equation}

which combines the stronger near-interface field of finer morphology with its shorter penetration depth. For a general ellipsoid, mapping the residual shape to a sphere gives the equivalent modal form

\begin{equation}
\|h\|_{X_E}^2
\sim
\sum_{\ell\ge0}(\ell+1)A_\ell^2 ,
\tag{11}
\end{equation}

where \(A_\ell\) is the RMS amplitude of angular degree \(\ell\). The order \(1/2\) is not only sufficient. For smooth increasingly fine perturbations measured in an \(H^s\)-type norm,

\begin{equation}
\frac{q(\Omega_{h_\ell})}{\rho_{E,s}(h_\ell)}
\asymp
\ell^{1/2-s}.
\tag{12}
\end{equation}

Hence a frequency-independent upper bound fails for \(s<1/2\), while a positive frequency-independent lower bound fails for \(s>1/2\). The complete optimal-order argument is given in \textbf{SI} Proposition~\ref{SI-pub:prop-sharpness}, with the tangent-plane calculation in Section~\ref{SI-pub:plane}.

The spherical case makes this scale dependence explicit. After the best ellipsoidal directions are removed, write the residual height in spherical harmonics of degree \(\ell\ge3\). Degrees 0, 1, and 2 represent scaling, translation, and infinitesimal ellipsoidal changes. The exact linear spectrum is

\begin{equation}
\boxed{
q_{\mathrm{lin}}^2
=
3\sum_{\ell=3}^{\infty}
\frac{\ell(\ell-1)}{2\ell+1}A_\ell^2 .
}
\tag{13}
\end{equation}

At fixed height RMS, finer modes give larger bulk heterogeneity; for example, degree 12 produces about 2.48 times the RMS response of degree 3 at the same height. At fixed RMS slope, however, the amplitude of fine modes must decrease roughly as \(1/\ell\), reversing the ranking. There is therefore no wavelength-independent statement that fine or broad imperfections are intrinsically more severe without specifying the geometric constraint.

Using the half-order distance

\begin{equation}
\rho_{\mathrm{lin}}^2
=
\sum_{\ell\ge3}(\ell+1)A_\ell^2 ,
\tag{14}
\end{equation}

both geometry and mechanical response scale as \(\ell A_\ell^2\) at high degree, yielding the sharp all-mode bounds

\begin{equation}
\boxed{
\frac{3}{\sqrt{14}}\rho_{\mathrm{lin}}
\le
q_{\mathrm{lin}}
\le
\sqrt{\frac32}\rho_{\mathrm{lin}} .
}
\tag{15}
\end{equation}

Figure~\ref{fig:scale-spectrum}B compares the same linearized response divided by three different geometric measures, rather than three distinct responses of one shape. The half-order metric does not make all modes mechanically identical; it removes the divergent wavelength dependence that prevents two-sided control. It remains geometry-only rather than embedding the elastic spectrum into the metric. For sufficiently small finite perturbations, the modal prediction acquires the controlled nonlinear correction of Eq. (7), linking this explicit sphere result back to the general theorem. The complete spherical derivation is given in \textbf{SI} Section~\ref{SI-pub:sphere}; the finite-amplitude checks are summarized in Methods Section~\ref{sec:method-sphere} and detailed in \textbf{SI} Section~\ref{SI-pub:mode-checks}.

\begin{figure}[!htbp]
\centering
\includegraphics[width=\linewidth,keepaspectratio]{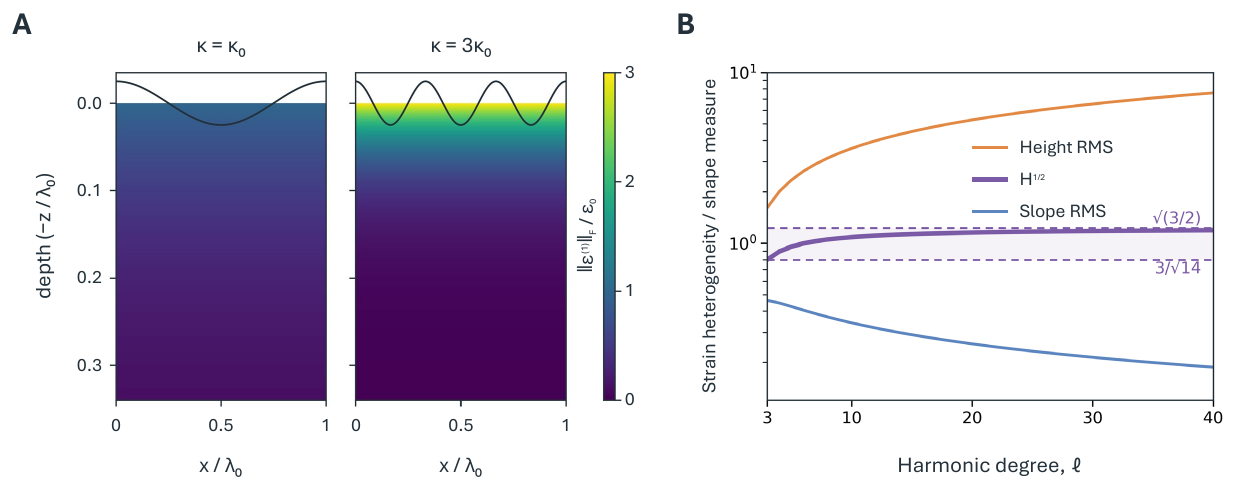}
\caption[Morphological scale and half-order control of strain heterogeneity.]{\input{sections/caption_2}}
\label{fig:scale-spectrum}
\end{figure}

\FloatBarrier

\subsection{Finite-shape stability of Eshelby uniformity}
\label{sec:finite-shape-stability}
Section 2.1 gives a rigorous local stability statement over an infinite-dimensional neighborhood of the ellipsoid family. A complementary question arises for a prescribed finite geometry that may lie outside that neighborhood: how tightly can its loss of uniformity be bounded, and can those bounds converge to the exact response? This second stability problem uses the same exact elasticity-to-potential reduction but does not require the local theorem.

Choose an ellipsoid \(F\) containing \(\Omega\) (Fig.~\ref{fig:overview}D), with \(V=|\Omega|\) and \(W=|F|\), and define the volume-balanced residual source

\begin{equation}
g=\chi_\Omega-\frac{V}{W}\chi_F .
\tag{16}
\end{equation}

Here \(\chi_D\) denotes the indicator function of a region \(D\): it equals one inside \(D\) and zero outside.

Because an ellipsoid has a constant potential Hessian throughout its interior, subtracting the scaled ellipsoidal source changes only the uniform field inside \(\Omega\) and leaves \(q\) unchanged. Define

\begin{equation}
\mathbf G=\operatorname{dev}\nabla^2N_g .
\tag{17}
\end{equation}

The volume balance makes the total squared residual field known exactly:

\begin{equation}
\boxed{
\|\mathbf G\|_{L^2(\mathbb R^3)}^2
=
\frac{2V}{3}\left(1-\frac{V}{W}\right)
\equiv\mathcal B .
}
\tag{18}
\end{equation}

Thus, before resolving the detailed field, the geometry fixes a global residual-field budget \(\mathcal B\). The key step is to partition this budget according to mechanical role,

\begin{equation}
\boxed{
\mathcal B
=
Vq^2
+
\frac{\left|\int_\Omega\mathbf G\,dV\right|_F^2}{V}
+
\int_{\Omega^c}|\mathbf G|_F^2\,dV .
}
\tag{19}
\end{equation}

The three terms are the desired internal heterogeneity, the squared internal mean residual field, and residual energy outside the inclusion.

This exact accounting identity produces two-sided bounds from partial field information. For any resolved interior core \(S\subset\Omega\), let \(A_S=\int_S|\mathbf G|_F^2dV\), \(\mathbf J_S=\int_S\mathbf G\,dV\), and \(v_S=|S|\). Then

\begin{equation}
\boxed{
Vq^2
\ge
A_S-\frac{|\mathbf J_S|_F^2}{v_S}.
}
\tag{20}
\end{equation}

For a resolved exterior region \(T\subset\Omega^c\), let \(E_T=\int_T|\mathbf G|_F^2dV\), \(v_R=V-v_S\), and \(R_E=\mathcal B-A_S-E_T\). Cauchy-Schwarz bounds the unresolved contribution to the internal mean, giving

\begin{equation}
\boxed{
Vq^2
\le
\mathcal B-E_T
-
\frac{1}{V}
\max\!\left(
0,
|\mathbf J_S|_F-\sqrt{v_RR_E}
\right)^2 .
}
\tag{21}
\end{equation}

The construction is morphology-independent: protrusions, indentations, necks, or more complicated prescribed shapes enter through the same residual source and field integrals. More importantly, the gap is not intrinsic. For nested regions \(S_N\uparrow\Omega\) and \(T_N\uparrow\Omega^c\), where \(\uparrow\) denotes increasing sequences that exhaust the respective domains up to sets of zero volume, let \(\ell_N\) and \(u_N\) denote the raw lower and upper endpoints from Eqs. (20) and (21). Exact integration gives \(\ell_N\to q\) and \(u_N\to q\). Defining the cumulative envelopes \(L_N=\max_{k\le N}\ell_k\) and \(U_N=\min_{k\le N}u_k\) makes the tightening monotone:

\begin{equation}
\boxed{
L_N\uparrow q,
\qquad
U_N\downarrow q .
}
\tag{22}
\end{equation}

The lower bounds therefore increase and the upper bounds decrease toward \(q\). The finite-shape theory is asymptotically exact for a fixed geometry. Section 2.1 supplies universal all-mode control near ellipsoids; Section 2.3 sacrifices that universality to obtain systematically improvable finite-amplitude bounds for one geometry. At the continuum level Eqs. (18)--(22) are exact; their complete derivation is given in \textbf{SI} Theorem~\ref{SI-pub:thm-finite} and Section~\ref{SI-pub:finite-convergence}. The reported numerical endpoints below are working numerical evaluations with fixed-region quadrature checks; full computer-assisted certification additionally requires outward numerical enclosure of the field integrals. The application geometries and numerical enclosure protocol are summarized in Methods Section~\ref{sec:method-finite-shape} and detailed in \textbf{SI} Sections~\ref{SI-pub:geometries}--\ref{SI-pub:intervals}.

\subsection{Finite-shape applications and finite-element comparison}
\label{sec:finite-shape-validation}
We apply the finite-shape theory to controlled morphologies motivated by irregular and coalescing precipitates: real precipitates can depart substantially from simple spheroids and develop interconnected or necked configurations during coalescence, while ellipsoidal fitting remains a useful reduced geometric description \cite{Marquis2005,Karnesky2007,Vogel2015}. A prolate spheroid of aspect ratio 1.5 provides the exact reference (Case A, \(q=0\)); \(R_0=(3V_A/4\pi)^{1/3}\) is the radius of the sphere having its volume \(V_A\). The protrusion (Case B) and indentation (Case C) are constructed by adding and subtracting the same axial height profile over the same polar footprint. Their maximum axial height/depth is \(0.35R_0\). The necked particle (Case D) is an equal-volume connected two-lobe shape with neck-to-lobe radius ratio \(r_0/r_1=0.60\), not treated as a small perturbation of one ellipsoid. Figure~\ref{fig:finite-validation}A presents these geometries in the order protrusion, indentation, and necked particle. The analytic profiles and containing ellipsoids are specified in Methods Section~\ref{sec:method-finite-shape} and \textbf{SI} Section~\ref{SI-pub:geometries}.

All three cases use the same inequalities and refinement rule. The inclusion and its containing ellipsoid remain fixed. At refinement level \(N=0,\ldots,4\), the omitted volume fraction is \(\tau_N=0.1\,2^{-N}\) on each side of the interface, measured relative to the inclusion volume \(|\Omega|\). The interior core is \((1-\tau_N)^{1/3}\Omega\) and the resolved exterior begins outside \((1+\tau_N)^{1/3}\Omega\). Each level halves the omitted volume on either side of the interface. The complete far field is included through a surface integral on a fixed enclosing sphere; fixed-region quadrature checks are performed separately from this geometric refinement. At the finest level, \(N=4\) and \(\tau_4=0.625\%\), the working numerical intervals are

\begin{equation}
\boxed{\begin{aligned}
0.051989&\le q_B\le0.082406,\\
0.075409&\le q_C\le0.077286,\\
0.154127&\le q_D\le0.167759.
\end{aligned}}
\tag{23}\label{eq:finite-intervals}
\end{equation}

The interval widths in \(q\) decrease by \(81.3\%\), \(91.6\%\), and \(88.5\%\) from the first to the final level, respectively (Fig.~\ref{fig:finite-validation}C). The same omitted-volume fractions make the refinement comparable across geometries, without assuming that equal-volume layers contain equal field energy. The gaps \(u_N^2-\ell_N^2\) show contraction approaching first order over the computed range; the observed orders and numerical checks are reported in \textbf{SI} Sections~\ref{SI-pub:dyadic}--\ref{SI-pub:intervals}.

We use the finite element method (FEM) to compute \(q\) independently for the three geometries. All three values lie within the predicted intervals. These displacement-based calculations solve the equilibrium equations,

\begin{equation}
\nabla\cdot\boldsymbol{\sigma}=0,
\qquad
\boldsymbol{\sigma}=\mathbb C:\left(\boldsymbol{\varepsilon}-\alpha\mathbf I\chi_\Omega\right).
\tag{24}\label{eq:fem-equilibrium}
\end{equation}

Sphere and spheroid controls give residual \(q\) values of \(2.76\times10^{-5}\) and \(2.10\times10^{-5}\), establishing the numerical floor. The FEM responses, with the mesh and domain checks reported in \textbf{SI} Section~\ref{SI-pub:fem-checks}, are

\begin{equation}
q_B^{\rm FEM}=0.052435,
\qquad q_C^{\rm FEM}=0.076254,
\qquad q_D^{\rm FEM}=0.154657.
\tag{25}\label{eq:fem-values}
\end{equation}

At the finest level, the lower endpoints reach \(99.2\%\), \(98.9\%\), and \(99.7\%\) of those reference values. The protrusion retains the largest relative upper allowance; the indentation is tightly bracketed, and the necked particle has a smaller but still visible upper allowance. This asymmetry is consistent with the bound construction: the lower endpoint measures resolved interior variation, whereas the upper endpoint also depends on exterior information and the estimate of the interior mean. The FEM formulation and convergence checks are summarized in Methods Section~\ref{sec:method-fem} and \textbf{SI} Sections~\ref{SI-pub:fem}--\ref{SI-pub:fem-checks}.

The FEM fields provide the mechanical interpretation (Fig.~\ref{fig:finite-validation}B). The protrusion produces variation around its polar feature and attachment, the indentation concentrates variation around its rim, and the necked particle develops a broader heterogeneous region through the constriction. Adding and removing the same-sized feature therefore give different responses, while the volume-preserving necked shape produces the largest heterogeneity. Changed volume or fitted aspect ratio alone cannot quantify loss of Eshelby uniformity; the spatial redistribution of the elastic field matters.

For specified \(\alpha\), \(K\), and \(\mu\), each dimensionless interval gives RMS strain and stress fluctuation estimates through \(|\alpha|\beta q\) and \(2\mu|\alpha|\beta q\).

For a fixed inclusion geometry, the continuum bounds converge to the exact heterogeneity as progressively more of the interior and exterior is resolved. Their gap can therefore be made arbitrarily small, allowing the refinement level to be chosen according to the accuracy required by an engineering application. For illustration, consider a normalized RMS allowance \(q_{\mathrm{allow}}=0.08\). The evaluated upper bound for the indentation (Case C) lies below this allowance, while the lower bound for the necked particle (Case D) exceeds it. The protrusion interval (Case B) straddles the allowance, motivating further refinement of the bounds for this same geometry.

\begin{figure}[!htbp]
\centering
\includegraphics[width=0.97\linewidth,keepaspectratio]{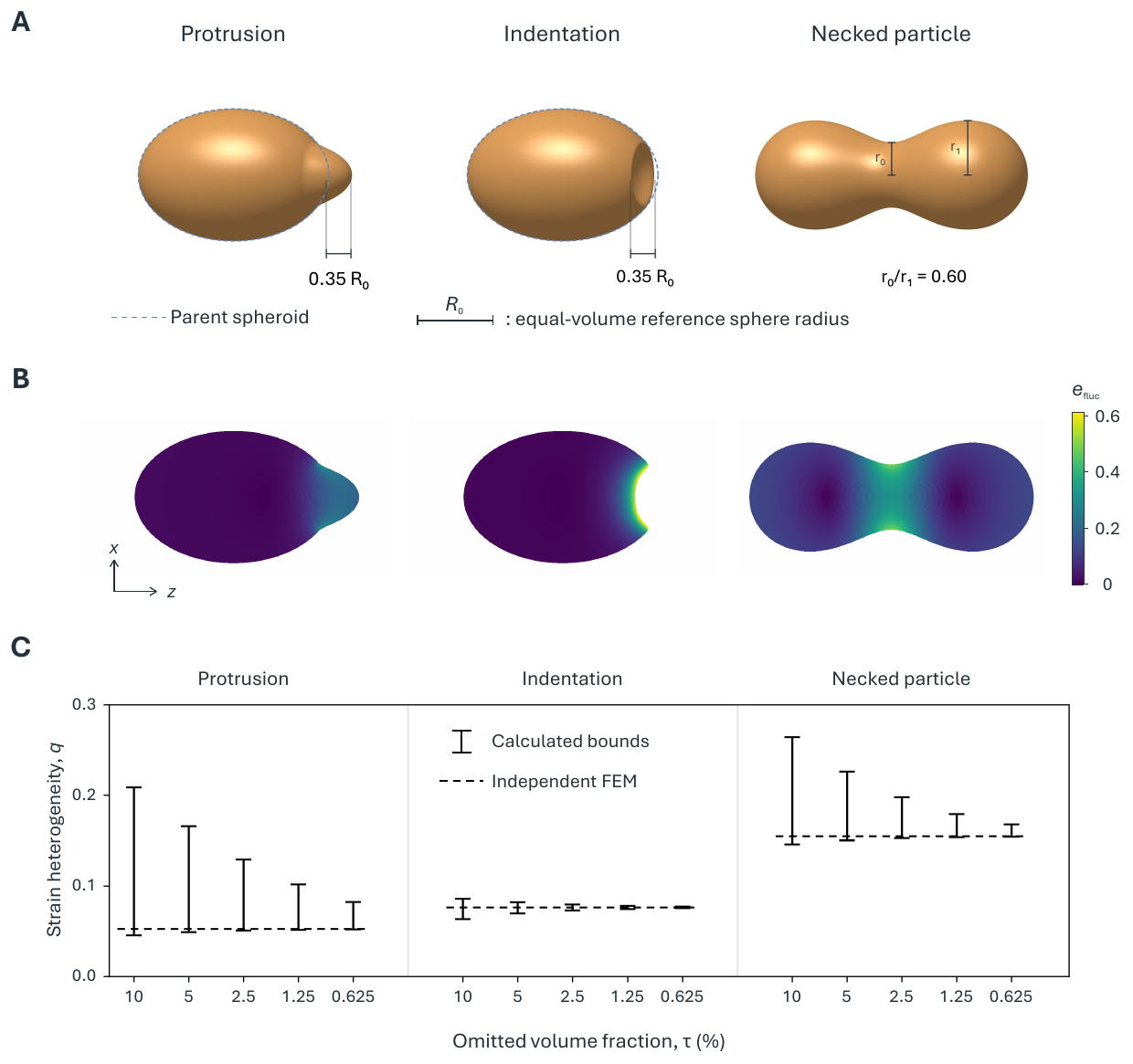}
\caption[Finite-shape bounds and independent finite-element comparison.]{\input{sections/caption_3}}
\label{fig:finite-validation}
\end{figure}

\FloatBarrier

\section{Discussion}
\label{sec:discussion}
\subsection{Two notions of Eshelby stability}
\label{sec:two-stabilities}

The two theoretical results address different notions of stability. Section 2.1 provides two-sided control of strain heterogeneity throughout an admissible class of shapes near any fixed nondegenerate ellipsoid. The same constants apply to all admissible perturbations and their combinations, with geometric departure measured relative to the nearby ellipsoid family. The half-order shape distance captures both perturbation amplitude and spatial scale, providing a basis for geometric specifications that limit or generate strain heterogeneity across this entire class. It should be noted that this correspondence is local in shape space. Related two-dimensional inclusion examples show that almost uniform fields can also occur in strongly nonelliptical shapes \cite{KangKimMilton2008}. The finite-shape results provide a complementary route for prescribed geometries beyond this neighborhood, with refinable bounds on their strain heterogeneity.

For a prescribed inclusion geometry, the finite-shape results provide a quantitative basis for assessing whether its strain heterogeneity meets an engineering tolerance. The exact elastic reduction and residual-field budget give upper and lower bounds that converge to the exact heterogeneity as interior and exterior field information is resolved. Provided the geometry is known to arbitrary accuracy, these continuum bounds can be refined to any prescribed accuracy for the same geometry, allowing further refinement when an initial interval is inconclusive. For shapes approaching an ellipsoid, the finite-shape response is consistent with the local asymptotic theory.

The applications show why both viewpoints are useful. Cases B and C have the same maximum height/depth and footprint but different heterogeneity, while the necked shape D preserves the reference volume and produces the largest response. Volume change, maximum height, or fitted aspect ratio therefore cannot by themselves quantify loss of uniformity. The observable \(q\) gives this comparison a mechanical meaning by measuring the normalized RMS error of the best uniform-strain approximation within the inclusion. Under the present assumptions, it also determines RMS stress fluctuations and the fluctuation contribution to the interior elastic energy, \(U_{\mathrm{fluct}}=\mu|\Omega|(|\alpha|\beta)^2q^2\) (\textbf{SI} Section~\ref{SI-pub:mechanics}). The bounds therefore support engineering tolerances for these quantities and for criteria quantitatively related to them. Further estimates for peak stresses and for the difference between the actual mean strain and the uniform strain predicted by a fitted ellipsoid would address additional engineering criteria.

Together, these results connect Eshelby's exact uniform-field property to both rigorous class-wide control and engineering-level prediction. The distinction also illustrates a broader feature of mechanics: general theory, shape-specific computation, and independent computational comparison answer different questions, and their value is greatest when they are designed to reinforce rather than substitute for one another.

\subsection{Human-led, AI-assisted theoretical mechanics}
\label{sec:ai-mechanics}

Recent agentic systems in mechanics and materials have combined specialized language-model agents with knowledge retrieval, coding, physics-based simulation, execution feedback, critique, and adaptive team organization across mechanics problem solving and materials-design workflows \cite{NiBuehler2024,AlloyGen2026,SciAgents2025,AtomAgents2025,ALLFEM2026}.

The present study explored a human-led, AI-assisted research structure in which scientific decision authority remained with the researchers, while AI enlarged the space of candidate reasoning (Fig.~\ref{fig:overview}B). Human researchers selected the mechanics problem, observables, assumptions, acceptance criteria, and validation requirements, and retained responsibility for interpretation and final claims. AI supported literature and cross-domain search, alternative formulations and derivations, symbolic and numerical checks, coding, partial formalization, and parallel exploration. AI-generated results were treated as candidates until reviewed through complementary checks designed to expose different failure modes (Table 1); the research process and AI review protocol are summarized in Methods Section~\ref{sec:method-ai}, with supporting details in \textbf{SI} Section~\ref{SI-pub:process}.

\begin{table}[htbp]
\centering
\small
\caption{Roles, authority, and checks in the human-led, AI-assisted research structure.}
\setlength{\tabcolsep}{3pt}
\begin{tabular}{@{}>{\raggedright\arraybackslash}p{0.18\linewidth}>{\raggedright\arraybackslash}p{0.25\linewidth}>{\raggedright\arraybackslash}p{0.22\linewidth}>{\raggedright\arraybackslash}p{0.27\linewidth}@{}}
\toprule
Research function & Human role and authority & AI contribution & Review before acceptance \\
\midrule
Problem formulation & Select the question, observable, assumptions, and target; decide when the formulation must change & Survey literature and related mathematics; propose alternatives & Physical consistency, limiting cases, and known results \\
Theory development & Maintain physical meaning and research direction; reject insufficient routes & Develop derivations, explore parallel routes, and perform symbolic/numerical tests & Human audit; cross-checks; staged reconstruction by a fresh-context AI reviewer given the definitions and target but not the prior conversation or candidate proof \\
Verification and validation & Define what counts as independent evidence and decide whether a claim is accepted & Reproduce calculations, test identities, support formalization and numerical implementation & Partial Lean formalization; independent numerical routes; displacement-based FEM independent of the potential formulation \\
Interpretation & Determine mechanics significance, scope, limitations, and final claims & Organize results and compare alternatives & Human final audit against the physical question and evidence \\
\bottomrule
\end{tabular}
\end{table}

Several episodes show why this structure mattered. Measuring shape error from one reference ellipsoid was rejected because another exact ellipsoid can differ geometrically while retaining uniformity; the target therefore became distance from the ellipsoid family. Sharp linearized gains were not promoted directly to finite-amplitude bounds when finite-shape calculations exposed nonlinear corrections. Agreement between two implementations sharing the same potential representation was treated as a consistency check rather than independent evidence, motivating the displacement-based FEM test. These cases illustrate that AI was most useful for broadening and sustaining reasoning, while human judgment determined which question mattered, which result was insufficient, and what evidence was required before acceptance. This separation between generative breadth and scientific acceptance is important more generally: greater AI-enabled productivity does not by itself guarantee scientific understanding, while effective use in materials research still depends on domain grounding, hypothesis testing, and explicit evaluation aligned with scientific goals \cite{MesseriCrockett2024,MiretKrishnan2025,ZhangEtAl2025ScientificMethod}. This research configuration is one explored case rather than a fixed prescription; the next section considers how such effort might be directed toward a broader set of mechanics problems.

\subsection{From one problem to a mechanics problem map}
\label{sec:problem-map}

Mechanics is a particularly demanding setting for open-problem formulation because a useful question must connect a physical phenomenon or engineering challenge, a mathematical model, a precise unresolved statement, and the physical or engineering implications of resolving it. This coupling is also what makes mechanics attractive for increasingly capable AI-assisted reasoning. More analytical search or computation is valuable only if the problem being explored remains physically meaningful and if progress can be distinguished from reformulation, approximation, or numerical evidence. The present study illustrates this directly: substantial effort was required not only to derive results, but to decide what quantity should be controlled, how shape variation should be measured relative to the ellipsoid family, when linearized information was insufficient, and what constituted independent evidence.

To make such boundaries explicit beyond one problem, we developed \emph{Open Problems in Mechanics} (OPM) (Fig.~\ref{fig:overview}E), a versioned public map of selected theory-rich questions in solid mechanics, structural mechanics, and mechanics of materials. The collection is publicly available at \url{https://generativemechanicslab.github.io/mechanics-open-problems/} \cite{Ni2026OPM}. OPM welcomes corrections and extensions to individual entries, as well as proposals for new problems, from the mechanics community. Its initial version-1.0 collection contains ten problems (Table 2). Problems were retained when they combined a substantial theoretical foundation, relevance across a broader family of physical or engineering systems, a genuinely unresolved core after literature review, and a target precise enough that progress could be recognized. Each entry records the physical question, mathematical or theoretical kernel, what is established, recent progress, the current known/open boundary, and what would constitute meaningful progress or settlement. Problems were revised, superseded, or held out when the residual target could not yet be stated sharply. The curation and versioning procedure is summarized in Methods Section~\ref{sec:method-opm}. Building on this curation and review, we plan to invite topic specialists to examine individual entries in greater depth, refine their formulation and literature coverage, and identify concrete next milestones. Publication status, including preprints and peer-reviewed papers, will be recorded separately from the scope of independent verification.

\begin{table}[htbp]
\centering
\small
\caption{First-release problem map in \emph{Open Problems in Mechanics}.}
\setlength{\tabcolsep}{3pt}
\begin{tabular}{@{}>{\raggedright\arraybackslash}p{0.08\linewidth}>{\raggedright\arraybackslash}p{0.59\linewidth}>{\raggedright\arraybackslash}p{0.25\linewidth}@{}}
\toprule
ID & Open question & Mechanics theme \\
\midrule
EL01 & How stable is Eshelby uniformity to shape imperfections? & Inclusion mechanics; quantitative stability \\
FR01 & What determines the attainable speed and first instability of a tensile crack? & Dynamic fracture \\
DI02 & How does continuum plasticity emerge from collective dislocation dynamics? & Dislocations; multiscale plasticity \\
PL01 & What selects the physical evolution after plasticity loses ellipticity? & Localization; post-instability selection \\
NL01 & Which singularities are physically admissible in finite elasticity? & Nonlinear elasticity; singularities \\
SH01 & What sets the safe buckling load of an imperfect cylindrical shell? & Structural stability; imperfection sensitivity \\
CO01 & When does three-dimensional elastic contact with Coulomb friction admit a quasistatic solution? & Contact and friction \\
S01 & What randomness survives in macroscopic fracture of a random solid? & Stochastic fracture; homogenization \\
AM02 & Can a Griffith-type fracture law be defined for a solidifying metal? & Solidification fracture; manufacturing \\
BIO01 & How does structural hierarchy transform local failure into macroscopic fracture resistance? & Biomechanics; hierarchical fracture \\
\bottomrule
\end{tabular}
\end{table}

The format draws on a long mathematical tradition in which explicitly formulated problem lists help organize unresolved questions and define recognizable targets for progress, from Hilbert's problems to later curated lists such as Smale's \cite{Hilbert1902,Smale1998}; modern resources continue this practice through maintained, updateable collections \cite{AIMProblemLists2026,OpenProblemGarden2026}. More recently, AI-assisted mathematical systems have produced new results on established problems by coupling generative search with explicit evaluators or formal verification, reinforcing the importance of separating candidate generation from verification \cite{RomeraParedes2024,Hubert2026,Georgiev2025}. The 2026 forced Navier--Stokes blow-up development and subsequent human-readable reconstruction provide a recent example in which result generation, verification, and readable mathematical synthesis became distinct stages \cite{OpenAI2026NavierStokes,LeiRen2026}. Mechanics requires an additional layer: the theoretical statement must remain tied to the physical model and to the engineering meaning of its resolution. The present Eshelby problem provides one example. This work advances quantitative stability in a controlled isotropic, matched-modulus setting, while broader questions involving general eigenstrains, mismatch, anisotropy, interfaces, multiple inclusions, and larger finite departures remain open. OPM records such progress by moving the documented boundary of a problem rather than treating a partial result as closure. We therefore view the map as a practical interface between human problem selection, established mechanics knowledge, and expanding AI-assisted reasoning, whose usefulness will depend on continued critical updating by the mechanics community.

\section{Conclusions}
\label{sec:conclusions}
Eshelby’s classical result shows how a special geometry can reduce a nontrivial elastic field problem to a uniform interior response. Here we asked what remains when that geometry is imperfect. Near the ellipsoid family, we establish a quantitative two-sided relation between nonellipsoidal shape variation and interior strain heterogeneity, with a half-order geometry scale that makes morphological length scale part of the mechanics. The spherical spectrum makes this dependence explicit, while the finite-shape theory carries the same question to prescribed nonellipsoidal morphologies. Together, these results move Eshelby theory from exactness to robustness.

That robustness also acquires a practical meaning. By separating the ellipsoidal uniform-field contribution from the residual geometry, the finite-shape formulation gives tightening bounds on the heterogeneity generated by a known morphology, consistent with independent displacement-based FEM. The relevant engineering question is therefore not only whether a shape is close to an ellipsoid, but how much strain or stress variation an ellipsoidal approximation neglects. The present results also point to a broader future program: testing quantitative Eshelby stability for general eigenstrain modes, elastic mismatch, and anisotropy, where the equivalent-inclusion and uniformity structures already become substantially richer \cite{BarnettCai2018,YuanHuangWang2022}; for interfaces and interacting inclusions, where the geometries supporting uniform internal fields can change qualitatively \cite{DaiGao2025,DaiGaoRu2015}; and for larger finite departures from ellipsoidal geometry. Extensions to finite-deformation and nonlinear inclusion mechanics provide a further direction \cite{LiEtAl2022}.

The study further provides one documented practice of human-led, AI-assisted theoretical mechanics. AI broadened analytical and computational exploration, while human researchers retained authority over the physical question, acceptance criteria, independent evidence, and interpretation. To extend this practice beyond the present case, we developed \emph{Open Problems in Mechanics} as a public, versioned map of selected theory-rich questions with their physical motivation, established results, known/open boundaries, and criteria for progress made explicit. The aim is not to define a closed list of grand challenges, but to provide a shared resource that can be corrected, sharpened, and advanced by the mechanics community. As AI expands the scale of available reasoning, making consequential questions and evidentiary standards explicit may become increasingly important for turning that capacity into mechanics advances that remain physically meaningful, checkable, and transferable.

\section{Materials and Methods}
\label{sec:methods}

\subsection{Mechanical formulation and local stability theorem}
\label{sec:method-local}

The elastic tensor is identical throughout \(\mathbb R^3\) and acts on symmetric tensors as
\begin{equation*}
\mathbb C\mathbf A=2\mu\mathbf A+\lambda\,\mathrm{tr}(\mathbf A)\mathbf I,
\qquad
\mu>0,\quad K=\lambda+2\mu/3>0 .
\end{equation*}
The eigenstrain is \(\alpha\mathbf I\chi_\Omega\), and the disturbance displacement belongs to \(\mathcal D^{1,2}(\mathbb R^3;\mathbb R^3)\). Weak equilibrium gives the Newtonian representation in Eqs. (1)--(2), and Eq. (3) defines the normalized whole-interior heterogeneity.

The main text fixes a \(C^{1,1/2}\)-type normal-graph norm \(Z_E\) for geometric admissibility; the \textbf{SI} proves the theorem for every fixed \(C^{1,\gamma}\), \(0<\gamma<1\). The shape error itself is measured in the geometry-only \(H^{1/2}\)-type norm \(X_E\). The proof constructs a whole-space coordinate lift, differentiates the shape-to-field map in the energy space, identifies the nine ellipsoidal null directions, proves transverse upper/lower control, fits an actual nearby ellipsoid, and bounds the registered nonlinear remainder by \(C_{\mathrm{rem},E}\delta\rho_E(h)\). The complete proof and its analytical inputs are given in \textbf{SI} Sections~\ref{SI-pub:setting}--\ref{SI-pub:scale-spectrum}; the relative nonlinear estimate is Eq.~\eqref{SI-pub:eq-relative-remainder}.

\subsection{Spherical spectrum and finite-amplitude checks}
\label{sec:method-sphere}

Real spherical harmonics are RMS-orthonormal on \(S^2\). Differentiating the linearized Newtonian single-layer potential and integrating over the ball gives Eq. (13); the complete spectral derivation is given in \textbf{SI} Section~\ref{SI-pub:sphere}. Symbolic checks verify the constant-Hessian subtraction, harmonic orthogonality, and modal coefficients. Selected finite-amplitude radial graphs, together with exact sphere and spheroid controls and a separate direct surface-integral implementation, test the asymptotic spectrum without entering the proof; these checks are detailed in \textbf{SI} Sections~\ref{SI-pub:potential-numerics}--\ref{SI-pub:mode-checks}.

\subsection{Finite-shape enclosure calculation}
\label{sec:method-finite-shape}

For each prescribed inclusion \(\Omega\), a containing ellipsoid \(F\) is selected from geometry alone and then held fixed. Cases B and C add and subtract the same compact axial height profile on the upper meridian of the aspect-ratio-1.5 reference spheroid, with maximum height/depth \(0.35R_0\) and footprint radius \(0.55b\); the opposite meridian is unchanged. Case D is a reference-volume-preserving connected two-lobe shape with \(r_0/r_1=0.60\). Here \(R_0\) is fixed by the unmodified reference volume. \textbf{SI} Section~\ref{SI-pub:geometries} specifies the analytic profiles and containing ellipsoids; Theorem~\ref{SI-pub:thm-finite} derives the bounds from \(A_S\), \(\mathbf J_S\), and \(E_T\).

A single parameter controls geometric refinement: \(\tau_N=0.1\,2^{-N}\), \(N=0,\ldots,4\), with \(S_N=(1-\tau_N)^{1/3}\Omega\) and \(T_N=\mathbb R^3\setminus\Omega_N^+\), where \(\Omega_N^+=(1+\tau_N)^{1/3}\Omega\). Each omitted interface layer has volume \(\tau_NV\). The exterior integral is split at the fixed sphere \(R_*=4R_0\); the complete far field is included by the Green-identity surface formula in \textbf{SI} Section~\ref{SI-pub:far-field}. This sphere is an integration surface, not a physical boundary condition.

Azimuthal integration is performed analytically using source-ring kernels expressed through complete elliptic integrals \cite[Eq.~(3.2) and Appendix~C]{ZhangZhouDing2023}; the remaining meridional and target integrals are evaluated numerically (\textbf{SI} Section~\ref{SI-pub:potential-numerics}). At fixed regions, quadrature sensitivity in either squared endpoint is compared with \(u_N^2-\ell_N^2\), using a 5\% acceptance threshold and an additional finest-level check. The largest recorded selected-pair change is 0.0506\% of that width. A separate \(4R_0\)-to-\(6R_0\) splitting-radius check is reported in \textbf{SI} Section~\ref{SI-pub:dyadic}. The theoretical endpoints are evaluated without using FEM results. The displayed intervals use the raw endpoints evaluated at each refinement level.

\subsection{Independent axisymmetric finite-element calculations}
\label{sec:method-fem}

The FEM solves Eq. (24) directly on the positive-\(r\), full-\(z\) meridian while retaining the three-dimensional axisymmetric strain state
\begin{equation*}
\varepsilon_{rr}=\partial_ru_r,\qquad
\varepsilon_{zz}=\partial_zu_z,\qquad
\varepsilon_{rz}=\frac12(\partial_zu_r+\partial_ru_z),\qquad
\varepsilon_{\theta\theta}=\frac{u_r}{r},
\end{equation*}
with volume weight \(2\pi r\). Six-node quadratic curved triangles fit the inclusion interface and use graded refinement in the matrix. On the rotation axis only \(u_r=0\) is imposed; \(u_z\) remains free. The distant spherical boundary is clamped, \(u_r=u_z=0\), and its radius is varied through \(8R_0\), \(16R_0\), and \(24R_0\), with the local mesh through \(3R_0\) unchanged during the domain study. Mesh refinements use 32/64/128 angular subdivisions and 6/12/24 interior radial layers, with an additional 256/48 refinement for Case C.

Sphere and spheroid controls test the eigenstrain implementation and finite-domain floor. At \(24R_0\), their full-field relative errors against the infinite-domain analytical solutions are approximately 0.00910\% and 0.00834\%, respectively. For Cases B, C, and D, the final mesh changes in \(q\) are approximately 0.015\%, 0.060\%, and 0.0025\%; changing the outer radius from \(16R_0\) to \(24R_0\) on the common fine mesh changes \(q\) by at most \(1.1\times10^{-5}\%\). The full Cartesian tensor mean is reconstructed before the Frobenius variance is evaluated. The FEM code contains no Newtonian-potential operator and does not use the theoretical interval endpoints. The formulation and convergence records are given in \textbf{SI} Sections~\ref{SI-pub:fem}--\ref{SI-pub:fem-checks}.

\subsection{AI-assisted research process and review}
\label{sec:method-ai}

Generative AI assistants (OpenAI, GPT-6 Astra) supported literature organization, candidate formulation, analytical exploration, symbolic and numerical checking, code development, partial formalization, and language editing. Human researchers determined the mechanics objectives, physical observables, scope, acceptance of assumptions, and validation requirements; conducted independent review and validation; and retained responsibility for the final interpretation.

The theorem, finite-shape enclosure, FEM calculation, and AI-process observations were treated as distinct evidence layers. Selected mathematical steps were reconstructed by an AI reviewer given the definitions and target before seeing the candidate proof. Partial Lean formalization checked selected mathematical components using the Lean 4 theorem prover \cite{deMouraUllrich2021}. Agreement between implementations sharing the same mathematical representation is treated as a consistency or verification check rather than as independent validation of that representation, consistent with established distinctions in computational verification and validation \cite{OberkampfTrucanoHirsch2004}. The displacement-based FEM calculation therefore provides a separate cross-formulation computational test of the finite-shape results. The separation among analytical proof, potential-based numerical checks, finite-shape bounds, and independent FEM is documented in \textbf{SI} Section~\ref{SI-pub:process}, which specifies the scope of the AI review and partial formalization.

\subsection{\emph{Open Problems in Mechanics} curation}
\label{sec:method-opm}

\emph{Open Problems in Mechanics} is maintained separately as a versioned problem map. Candidates were retained when they combined substantial theoretical structure, relevance across a broader family of physical or engineering systems, a genuinely unresolved core after literature review, and a target precise enough that progress could be recognized. Each page records the physical question, theoretical kernel, established results, recent progress, current known/open boundary, criteria for progress or settlement, references, and statement version. Candidates were revised, superseded, or held out when the residual target could not be stated sharply. Table 2 records the ten-problem first-release set used in this study.

\section*{Supporting information}
The Supplementary Information is included in this preprint immediately after the References. It contains the complete analytical support for the local and finite-shape results, the spherical-spectrum derivation, numerical evaluation details, independent finite-element comparison, and research-process documentation.

\section*{Data and code availability}
The \emph{Open Problems in Mechanics} collection and progress records are publicly available at \url{https://generativemechanicslab.github.io/mechanics-open-problems/}. Its structured problem records and website/maintenance code are available at \url{https://github.com/GenerativeMechanicsLab/mechanics-open-problems}. The analytical derivations, reported numerical tables, and convergence diagnostics for the Eshelby study are provided in the Supplementary Information.

\section*{Acknowledgments}
The author thanks Professor Huajian Gao at \mbox{Tsinghua University} for helpful discussions on the engineering interpretation of the results and the development of \emph{Open Problems in Mechanics}. The author also acknowledges OpenAI for providing complimentary access to ChatGPT through the ChatGPT for Academic Researchers programme.

\section*{Author contributions}
B.N. designed the study, conducted the analytical and computational investigations, analyzed the results, and wrote the paper.

\section*{Competing interests}
The author declares no competing interests.

\section*{Declaration of generative AI and AI-assisted technologies in the writing process}
The author used generative AI assistants (OpenAI, GPT-6 Astra) to assist with literature organization, candidate formulation, analytical exploration, symbolic and numerical checks, code development, partial formalization, and language editing. The author verified the reported results, reviewed and edited the manuscript, and takes full responsibility for its content.

\bibliographystyle{unsrtnat}
\bibliography{references}

\BeginArxivSupplement

\ArxivSupplementTitle

\section{Mechanical formulation and geometric setting}
\label{pub:setting}

\subsection{Equilibrium and the normalized strain fluctuation}
\label{pub:mechanics}
We use the observable \(q\) and material normalization of main-text Eqs.~(1)--(3). For a bounded inclusion \(\Omega\subset\mathbb R^3\) of positive volume, the eigenstrain is \(\alpha\mathbf I\chi_\Omega\), with \(\alpha\ne0\). The stiffness is identical in the two phases:
\[
\mathbb C\mathbf A=2\mu\mathbf A+\lambda\operatorname{tr}(\mathbf A)\mathbf I,
\qquad \mu>0,\qquad K=\lambda+2\mu/3>0.
\]
Vectors and second-order tensors are bold; \((D\mathbf u)_{ij}=\partial_j u_i\), \(\boldsymbol\varepsilon(\mathbf u)=\operatorname{sym}D\mathbf u\), and \(|\mathbf A|_F^2=\sum_{ij}A_{ij}^2\). With no remote loading, the displacement belongs to
\(\mathcal D^{1,2}=\{\mathbf u\in L^6(\mathbb R^3;\mathbb R^3):D\mathbf u\in L^2\}\), and satisfies
\begin{equation}
\int_{\mathbb R^3}\mathbb C\boldsymbol\varepsilon(\mathbf u):\boldsymbol\varepsilon(\boldsymbol\varphi)\,dV
=3K\alpha\int_\Omega\operatorname{div}\boldsymbol\varphi\,dV,
\qquad \boldsymbol\varphi\in\mathcal D^{1,2}.
\label{pub:eq-weak}
\end{equation}
A global displacement encodes perfect bonding; the distributional equilibrium equation supplies traction continuity. Integration by parts for compactly supported smooth fields, followed by density, gives
\[
\|\boldsymbol\varepsilon(\mathbf v)\|_2^2
=\tfrac12\|D\mathbf v\|_2^2+\tfrac12\|\operatorname{div}\mathbf v\|_2^2.
\]
Together with \(\mathbf A:\mathbb C\mathbf A=2\mu|\operatorname{dev}\mathbf A|_F^2+K(\operatorname{tr}\mathbf A)^2\), this proves coercivity. The right-hand side of \eqref{pub:eq-weak} is bounded in the gradient norm, so Lax--Milgram gives a unique finite-energy solution.

Define the positive Newtonian potential by
\begin{equation}
N_\Omega(\mathbf x)=\int_\Omega\frac{d\mathbf y}{4\pi|\mathbf x-\mathbf y|},
\qquad -\Delta N_\Omega=\chi_\Omega.
\label{pub:eq-potential}
\end{equation}
The distributional equation \(\mu\Delta\mathbf u+(\lambda+\mu)\nabla\operatorname{div}\mathbf u=3K\alpha\nabla\chi_\Omega\) is solved by
\begin{equation}
\mathbf u=-\alpha\beta\nabla N_\Omega,
\qquad \boldsymbol\varepsilon=-\alpha\beta D^2N_\Omega,
\qquad \beta=\frac{3K}{K+4\mu/3}.
\label{pub:eq-elastic-reduction}
\end{equation}
The Hessian multiplier is \(-\xi_i\xi_j/|\boldsymbol\xi|^2\). Plancherel and \(\sum_{ij}\xi_i^2\xi_j^2/|\boldsymbol\xi|^4=1\) give \(\|D^2N_\Omega\|_2^2=|\Omega|\). The gradient kernel is locally integrable and decays as \(|\mathbf x|^{-2}\); for bounded \(\Omega\), \(\nabla N_\Omega\in L^6\). Thus \eqref{pub:eq-elastic-reduction} is the unique solution of \eqref{pub:eq-weak}.

For \(\langle\mathbf H\rangle_\Omega=|\Omega|^{-1}\int_\Omega\mathbf H\,dV\), orthogonal projection onto constant tensors yields
\begin{equation}
q(\Omega)=\frac{\|\boldsymbol\varepsilon-\langle\boldsymbol\varepsilon\rangle_\Omega\|_{L^2(\Omega)}}{|\alpha|\beta\sqrt{|\Omega|}}
=\frac{1}{\sqrt{|\Omega|}}\min_{\mathbf M\in\operatorname{Sym}(3)}
\|D^2N_\Omega-\mathbf M\|_{L^2(\Omega)}.
\label{pub:eq-q}
\end{equation}
The minimum is attained at the volume mean. In \(\Omega\), \(\operatorname{tr}\boldsymbol\varepsilon=\alpha\beta\) is constant. Consequently,
\begin{equation}
\boldsymbol\sigma-\langle\boldsymbol\sigma\rangle_\Omega
=2\mu(\boldsymbol\varepsilon-\langle\boldsymbol\varepsilon\rangle_\Omega),
\quad
\varepsilon_{\rm rms}=|\alpha|\beta q,
\quad \sigma_{\rm rms}=2\mu|\alpha|\beta q.
\label{pub:eq-physical-rms}
\end{equation}
These are volume RMS fluctuations, not maximum-field estimates.

Write \(\overline{\boldsymbol\varepsilon}=\langle\boldsymbol\varepsilon\rangle_\Omega\) and \(\boldsymbol\varepsilon'=\boldsymbol\varepsilon-\overline{\boldsymbol\varepsilon}\). Since the strain trace is constant inside the inclusion, \(\operatorname{tr}\boldsymbol\varepsilon'=0\). The elastic energy associated with these fluctuations is therefore
\[
U_{\mathrm{fluct}}
=\frac12\int_\Omega
\boldsymbol\varepsilon':\mathbb C:\boldsymbol\varepsilon'\,dV
=\mu\int_\Omega|\boldsymbol\varepsilon'|_{\mathrm F}^{2}\,dV
=\mu|\Omega|(|\alpha|\beta)^2q^2.
\]
The cross term between \(\boldsymbol\varepsilon'\) and the mean elastic strain \(\overline{\boldsymbol\varepsilon}-\alpha\mathbf I\) integrates to zero because \(\int_\Omega\boldsymbol\varepsilon'\,dV=0\). Thus, the interior elastic energy separates into a mean-strain contribution and \(U_{\mathrm{fluct}}\).

\subsection{Normal graphs and the two geometric norms}
\label{pub:norms}
For the local theorem, lengths are divided by the equivalent-volume radius of a fixed nondegenerate ellipsoid and its center is placed at the origin. Write
\begin{equation}
\begin{gathered}
E=\{\mathbf x:\mathbf x^T\mathbf Q\mathbf x<1\}=\mathbf T B,
\qquad \mathbf Q=\mathbf Q^T>0,\\
\det\mathbf Q=1,\qquad \mathbf T=\mathbf Q^{-1/2},\qquad V_E=4\pi/3.
\end{gathered}
\label{pub:eq-ellipsoid}
\end{equation}
Here \(B\) is the unit ball. The reference normal is \(\mathbf n_E=\mathbf Q\mathbf x/|\mathbf Q\mathbf x|\). The inclusion \(\Omega_h\) has boundary
\begin{equation}
\partial\Omega_h=\{\mathbf x+h(\mathbf x)\mathbf n_E(\mathbf x):\mathbf x\in\partial E\}.
\label{pub:eq-normal-graph}
\end{equation}
For \(f_h(\boldsymbol\theta)=h(\mathbf T\boldsymbol\theta)\), let \(\widehat f_h(\mathbf y)=f_h(\mathbf y/|\mathbf y|)\) on \(\mathcal A=\{1/2\le|\mathbf y|\le3/2\}\). Fix \(0<\gamma<1\) and use
\begin{align}
\|h\|_{Z_{E,\gamma}}&=\|\widehat f_h\|_{L^\infty(\mathcal A)}+
\|D\widehat f_h\|_{L^\infty(\mathcal A)}+
[D\widehat f_h]_{C^{0,\gamma}(\mathcal A)},
\label{pub:eq-Z}\\
\|h\|_{X_E}^2&=\frac1{4\pi}\sum_{\ell\ge0}(\ell+1)
\|(f_h)_\ell\|_{L^2(S^2)}^2.
\label{pub:eq-X}
\end{align}
The H\"older seminorm in \eqref{pub:eq-Z} is the supremum of \(|D\widehat f_h(\mathbf y)-D\widehat f_h(\mathbf z)|/|\mathbf y-\mathbf z|^\gamma\) over distinct points of \(\mathcal A\). The annulus specifies a norm convention only. The spherical-harmonic weights in \eqref{pub:eq-X} are equivalent to \((1+\ell(\ell+1))^{1/2}\), and pullback by the fixed smooth map \(\mathbf T\) preserves Sobolev order. Thus \(X_E\) is a geometry-only \(H^{1/2}(\partial E)\) norm, with RMS normalization, and \(Z_{E,\gamma}\hookrightarrow X_E\). The main text uses \(Z_E=Z_{E,1/2}\).

The strong norm controls geometric admissibility, not the measured error. Sufficiently small \(\|h\|_{Z_{E,\gamma}}\) makes \eqref{pub:eq-normal-graph} an embedded boundary. No volume, centroid, symmetry, or convexity constraint is added to the theorem. Constants may depend on \(E,\gamma\), and these norm conventions, and are not uniform over degenerating ellipsoids.

Let \(\mathcal M_E\) denote the compact local chart of normal heights of actual ellipsoids constructed in Section~\ref{pub:fitting}. The distance in main-text Eq.~(5) is
\begin{equation}
\rho_E(h)=\min_{g\in\mathcal M_E}\|h-g\|_{X_E}.
\label{pub:eq-rho}
\end{equation}
Its minimizer need not be unique. The chart is local and contains all sufficiently nearby ellipsoids; \eqref{pub:eq-rho} is not an unrestricted global ellipsoid-fitting problem.

\section{Local nonlinear stability}
\label{pub:local}

\subsection{Statement and linear operator}
\label{pub:theorem-statement}
For a scalar boundary height \(k\), define
\begin{equation}
\begin{gathered}
S_Ek(\mathbf x)=\int_{\partial E}\frac{k(\mathbf y)}{4\pi|\mathbf x-\mathbf y|}\,dS_{\mathbf y},\\
\mathcal L_Ek=V_E^{-1/2}\Pi_ED^2S_Ek,
\qquad \Pi_E\mathbf H=\mathbf H-\langle\mathbf H\rangle_E.
\end{gathered}
\label{pub:eq-linear-operator}
\end{equation}
All norms of \(\mathcal L_E k\) are volume \(L^2(E)\) norms. Let \(\mathcal K_E\) be the nine-dimensional space
\begin{equation}
\mathcal K_E=\left\{\mathbf b\cdot\mathbf n_E(\mathbf x)-
\frac{\mathbf x^T\mathbf B\mathbf x}{2|\mathbf Q\mathbf x|}:
\mathbf b\in\mathbb R^3,\ \mathbf B=\mathbf B^T\right\},
\label{pub:eq-tangent-space}
\end{equation}
and let \(P\) be its \(X_E\)-orthogonal projection. Define
\begin{equation}
c_E=\inf_{\substack{Pk=0\\\|k\|_{X_E}=1}}\|\mathcal L_Ek\|_2,
\qquad
C_E=\sup_{\substack{Pk=0\\\|k\|_{X_E}=1}}\|\mathcal L_Ek\|_2.
\label{pub:eq-gains}
\end{equation}
The boundedness of the operator, equality \(\ker\mathcal L_E=\mathcal K_E\), and positivity of \(c_E\) are proved below.

\begin{theorem}[Local two-sided stability]
\label{pub:thm-local}
For each fixed ellipsoid \eqref{pub:eq-ellipsoid} and \(0<\gamma<1\), there exist \(\delta_0>0\) and \(A_{E,\gamma}<\infty\) such that
\begin{equation}
(c_E-A_{E,\gamma}\delta)\rho_E(h)
\le q(\Omega_h)\le
(C_E+A_{E,\gamma}\delta)\rho_E(h),
\qquad \|h\|_{Z_{E,\gamma}}\le\delta\le\delta_0.
\label{pub:eq-local-theorem}
\end{equation}
Here \(0<c_E\le C_E<\infty\), \(c_E-A_{E,\gamma}\delta_0>0\), and \(c_E,C_E\) are the optimal limiting gains as \(\delta\to0\). The constants are independent of perturbation frequency, eigenstrain amplitude, and the elastic moduli after the normalization \eqref{pub:eq-q}.
\end{theorem}
The proof registers both phases on a fixed domain, establishes linear control modulo \(\mathcal K_E\), fits an actual ellipsoid, and bounds the nonlinear error relative to the fitted residual. The final Jacobian comparison restores the physical volume average. Neither \(\delta_0\) nor \(A_{E,\gamma}\) is numerically evaluated.

In auxiliary estimates below, \(C\) denotes a finite constant that may change between occurrences and depends only on the fixed \(E,\gamma\), and norm conventions. The symbols \(c_E,C_E\) are reserved for \eqref{pub:eq-gains}.

\subsection{Registration and the fixed-domain equation}
\label{pub:registration}
\begin{lemma}[Two-norm coordinate extension]
\label{pub:lem-extension}
There is a linear extension \(h\mapsto\mathbf V_h\), with trace \(h\mathbf n_E\), satisfying
\begin{equation}
\|D\mathbf V_h\|_{L^2(\mathbb R^3)}\le C\|h\|_{X_E},
\qquad
\|\mathbf V_h\|_\infty+\|D\mathbf V_h\|_\infty
\le C\|h\|_{Z_{E,\gamma}}.
\label{pub:eq-lift-bounds}
\end{equation}
For sufficiently small \(h\), \(\boldsymbol\Phi_h=\operatorname{Id}+\mathbf V_h\) is a global bi-Lipschitz map taking \(E\) to \(\Omega_h\) and their exteriors to one another.
\end{lemma}
\begin{proof}
Prescribe \(\mathbf b_h(\boldsymbol\theta)=h(\mathbf T\boldsymbol\theta)\mathbf n_E(\mathbf T\boldsymbol\theta)\). Let \(\mathbf w_h\) be its componentwise Poisson extension to \(B\), and set
\begin{equation}
\mathbf W_h(\mathbf y)=
\begin{cases}
\displaystyle\int_{S^2}\frac{1-|\mathbf y|^2}{4\pi|\mathbf y-\boldsymbol\theta|^3}
\mathbf b_h(\boldsymbol\theta)\,dS_{\boldsymbol\theta},& |\mathbf y|<1,\\
|\mathbf y|^{-1}\mathbf w_h(\mathbf y/|\mathbf y|^2),& |\mathbf y|>1,
\end{cases}
\qquad \mathbf V_h(\mathbf x)=\mathbf W_h(\mathbf T^{-1}\mathbf x).
\label{pub:eq-lift}
\end{equation}
The traces coincide on the sphere. For a degree-\(\ell\) component, the interior and exterior radial factors are \(r^\ell\) and \(r^{-\ell-1}\). Green's identity and orthogonality therefore give
\begin{equation}
\int_{\mathbb R^3}|D\mathbf W_h|_F^2\,d\mathbf y
=\sum_{\ell\ge0}(2\ell+1)\|\mathbf b_{h,\ell}\|_{L^2(S^2)}^2.
\label{pub:eq-lift-energy}
\end{equation}
Partial harmonic sums extend this identity from smooth finite sums to \(H^{1/2}\) traces. Multiplication by the fixed smooth normal is bounded on \(H^{1/2}\): in its difference seminorm, the extra term contains the kernel \(|\boldsymbol\theta-\boldsymbol\omega|^{-1}\), which is integrable on the sphere. Changing variables by \(\mathbf T\) proves the first estimate in \eqref{pub:eq-lift-bounds}.

For a scalar trace \(b\), the interior normal derivative is
\[
\Lambda b(\boldsymbol\theta)=\frac1{2\pi}\operatorname{p.v.}\int_{S^2}
\frac{b(\boldsymbol\theta)-b(\boldsymbol\omega)}{|\boldsymbol\theta-\boldsymbol\omega|^3}
\,dS_{\boldsymbol\omega}.
\]
Subtracting the tangential linear term, whose symmetric integral vanishes, bounds the remaining integral by
\(C\|b\|_{C^{1,\gamma}}\int_0^2 r^{\gamma-1}\,dr\). Tangential derivatives are directly controlled. The maximum principle applied to the harmonic Cartesian derivatives bounds the interior gradient. Differentiating the Kelvin extension, with \(r=|\mathbf y|\), \(\mathbf e=\mathbf y/r\), and \(\mathbf z=\mathbf y/r^2\), gives
\[
D\mathbf W_h=-r^{-2}\mathbf w_h(\mathbf z)\otimes\mathbf e
+r^{-3}D\mathbf w_h(\mathbf z)(\mathbf I-2\mathbf e\otimes\mathbf e),
\]
which controls the exterior gradient and tail. These estimates prove the second inequality in \eqref{pub:eq-lift-bounds}. General \(C^{1,\gamma}\) traces follow by smooth approximation with uniformly bounded \(C^{1,\gamma}\) norms and convergence in \(C^{1,\gamma'}\), \(0<\gamma'<\gamma\); no density in the full strong H\"older norm is assumed.

The piecewise extension has matching traces, so its weak derivative has no interface measure. It belongs to \(\mathcal D^{1,2}\) and is globally Lipschitz. If \(a_h=\|D\mathbf V_h\|_\infty<1/4\), contraction solves \(\mathbf x=\mathbf y-\mathbf V_h(\mathbf x)\) for every \(\mathbf y\), and
\begin{equation}
(1-a_h)^3\le\det D\boldsymbol\Phi_h\le(1+a_h)^3.
\label{pub:eq-J-bounds}
\end{equation}
The boundary image is \eqref{pub:eq-normal-graph}; the homeomorphism preserves the bounded and unbounded complementary components. Hence it identifies both phases exactly.
\end{proof}

Introduce the coefficient fields and the registered potential gradient together:
\begin{equation}
\begin{gathered}
\mathbf D_h=D\mathbf V_h,\qquad \mathbf F_h=\mathbf I+\mathbf D_h,
\qquad J_h=\det\mathbf F_h,\\
\mathbf A_h=J_h\mathbf F_h^{-1}\mathbf F_h^{-T},\qquad
\mathbf Q_h=J_h\mathbf F_h^{-T},\qquad
\mathbf U_h=(\nabla N_{\Omega_h})\circ\boldsymbol\Phi_h.
\end{gathered}
\label{pub:eq-coefficients}
\end{equation}
Here \(\mathbf Q_h\) is a cofactor matrix, distinct from the fixed ellipsoid matrix \(\mathbf Q\). Distributionally, \(\mathbf w=\nabla N_{\Omega_h}\) obeys
\(\int D\mathbf w:D\boldsymbol\psi=-\int_{\Omega_h}\operatorname{div}\boldsymbol\psi\).
Using \(\boldsymbol\psi=\boldsymbol\varphi\circ\boldsymbol\Phi_h^{-1}\) yields
\begin{equation}
\int_{\mathbb R^3}(D\mathbf U_h\mathbf A_h):D\boldsymbol\varphi
=-\int_E\mathbf Q_h:D\boldsymbol\varphi,
\qquad \boldsymbol\varphi\in\mathcal D^{1,2}.
\label{pub:eq-fixed-domain}
\end{equation}
The transformation retains the interface source and the displayed matrix order. The weak chain rule and change of variables hold by bi-Lipschitz composition and energy-space approximation. Furthermore,
\(\mathbf z^T\mathbf A_h\mathbf z\ge(1-a_h)^3(1+a_h)^{-2}|\mathbf z|^2\), so the transformed problem is uniformly coercive.

For \(\mathbf Q(\mathbf D)=\det(\mathbf I+\mathbf D)(\mathbf I+\mathbf D)^{-T}\) and \(\mathbf A(\mathbf D)=\det(\mathbf I+\mathbf D)(\mathbf I+\mathbf D)^{-1}(\mathbf I+\mathbf D)^{-T}\), their derivatives at zero are
\begin{equation}
\mathbf Q_1(\mathbf D)=(\operatorname{tr}\mathbf D)\mathbf I-\mathbf D^T,
\qquad
\mathbf A_1(\mathbf D)=(\operatorname{tr}\mathbf D)\mathbf I-\mathbf D-\mathbf D^T.
\label{pub:eq-linear-coefficients}
\end{equation}
All second derivatives are bounded on a sufficiently small fixed matrix ball. For \(\mathbf Z=\mathbf D-\mathbf G\) and \(|\mathbf D|_F,|\mathbf G|_F\le t\), the fundamental theorem of calculus along the joining segment gives
\begin{equation}
\begin{aligned}
&|\mathbf Q(\mathbf D)-\mathbf Q(\mathbf G)-\mathbf Q_1(\mathbf Z)|_F
+|\mathbf A(\mathbf D)-\mathbf A(\mathbf G)-\mathbf A_1(\mathbf Z)|_F
\le Ct|\mathbf Z|_F,\\
&|(\mathbf I+\mathbf D)^{-1}-(\mathbf I+\mathbf G)^{-1}+\mathbf Z|_F
\le Ct|\mathbf Z|_F.
\end{aligned}
\label{pub:eq-coefficient-remainders}
\end{equation}
The corresponding first differences are bounded by \(C|\mathbf Z|_F\). No commutativity assumption is used.

\subsection{Energy derivative and the linear forward estimate}
\label{pub:linearization}
The classical ellipsoid formula gives a constant interior Hessian \(\mathbf H_E\) and bounded one-sided whole-space Hessians; the formula is specified in Section~\ref{pub:background} \cite{DiFratta2016}. Let \(\mathbf U_0=\nabla N_E\). For a direction \(k\), define \(\mathbf Y_k\in\mathcal D^{1,2}\) by
\begin{equation}
\int D\mathbf Y_k:D\boldsymbol\varphi
=-\int_E\mathbf Q_1(\mathbf D_k):D\boldsymbol\varphi
-\int D\mathbf U_0\mathbf A_1(\mathbf D_k):D\boldsymbol\varphi.
\label{pub:eq-derivative-equation}
\end{equation}
Whole-space integrals without a domain are over \(\mathbb R^3\). Testing with \(\mathbf Y_k\) and using \eqref{pub:eq-lift-bounds} gives
\begin{equation}
\|D\mathbf Y_k\|_2\le C\|\mathbf D_k\|_2\le C\|k\|_{X_E}.
\label{pub:eq-Y-bound}
\end{equation}
Subtracting \eqref{pub:eq-fixed-domain} at \(h=tk\) and \(0\), while keeping \(\mathbf A_{tk}\) on the left, first gives \(\|D(\mathbf U_{tk}-\mathbf U_0)\|_2\le C|t|\|\mathbf D_k\|_2\). After subtraction of \(t\) times \eqref{pub:eq-derivative-equation}, the residual sources are
\[
\begin{aligned}
&-\chi_E[\mathbf Q_{tk}-\mathbf I-t\mathbf Q_1(\mathbf D_k)],\\
&-D\mathbf U_0[\mathbf A_{tk}-\mathbf I-t\mathbf A_1(\mathbf D_k)],
\qquad -tD\mathbf Y_k(\mathbf A_{tk}-\mathbf I).
\end{aligned}
\]
Each is controlled using \eqref{pub:eq-coefficient-remainders}, giving
\begin{equation}
\|D(\mathbf U_{tk}-\mathbf U_0-t\mathbf Y_k)\|_2
\le C t^2\|\mathbf D_k\|_\infty\|\mathbf D_k\|_2.
\label{pub:eq-energy-derivative}
\end{equation}
This establishes the derivative in the full energy space.

On compact subsets of \(E\), differentiation of the moving-domain potential gives \(S_Ek\); differentiating the evaluation point contributes \(\mathbf H_E\mathbf V_k\). Identification with the energy derivative therefore yields
\begin{equation}
\mathbf Y_k=\nabla S_Ek+\mathbf H_E\mathbf V_k,
\qquad D^2S_Ek=D\mathbf Y_k-\mathbf H_E\mathbf D_k \quad\hbox{in }E.
\label{pub:eq-identify-derivative}
\end{equation}
The equality on every compact subset gives the distributional equality in \(E\), and hence \(\|D^2S_Ek\|_{L^2(E)}\le C\|k\|_{X_E}\). The gradient is controlled by \eqref{pub:eq-identify-derivative} and the Sobolev inequality. The potential itself obeys a kernel bound: the quantities \(\sup_{\mathbf x\in E}\int_{\partial E}|\mathbf x-\mathbf y|^{-1}dS\) and \(\sup_{\mathbf y\in\partial E}\int_E|\mathbf x-\mathbf y|^{-1}dV\) are finite, so weighted Cauchy--Schwarz controls \(\|S_Ek\|_{L^2(E)}\) by \(\|k\|_{L^2(\partial E)}\). Thus
\begin{equation}
\|S_Ek\|_{H^2(E)}+\|\mathcal L_Ek\|_2\le C\|k\|_{X_E}.
\label{pub:eq-forward}
\end{equation}
Density in \(X_E\), rather than in the strong H\"older norm, extends the operator to all half-order traces.

\subsection{Quantitative linear recovery and the ellipsoid kernel}
\label{pub:linear-recovery}
\begin{proposition}[Linear control transverse to ellipsoids]
\label{pub:prop-linear}
The kernel of \(\mathcal L_E\) is \(\mathcal K_E\), and
\begin{equation}
\operatorname{dist}_{X_E}(k,\mathcal K_E)\le C\|\mathcal L_Ek\|_2.
\label{pub:eq-inverse}
\end{equation}
Together with \eqref{pub:eq-forward}, this implies \(0<c_E\le C_E<\infty\) in \eqref{pub:eq-gains}.
\end{proposition}
\begin{proof}
Put \(u=S_Ek\in H^2(E)\) and \(\mathbf A=\langle D^2u\rangle_E\), which is trace-free because \(u\) is harmonic. Choose the linear and constant terms of \(p(\mathbf x)=\tfrac12\mathbf x^T\mathbf A\mathbf x+\mathbf b\cdot\mathbf x+c\) so that \(u-p\) and its gradient have zero mean. Two applications of Poincar\'e give
\begin{equation}
\|u-p\|_{H^2(E)}\le C\|\Pi_ED^2u\|_{L^2(E)}.
\label{pub:eq-quadratic-recovery}
\end{equation}
On the fixed smooth boundary, the Dirichlet and normal traces map \(H^2(E)\) into \(H^{3/2}(\partial E)\) and \(H^{1/2}(\partial E)\). If \(\mathcal D_E^+r\) denotes the outward-\(E\) normal derivative of the decaying exterior harmonic extension of \(r\), then
\begin{equation}
\|\mathcal D_E^+r\|_{H^{1/2}(\partial E)}\le C\|r\|_{H^{3/2}(\partial E)}.
\label{pub:eq-exterior-trace}
\end{equation}
For this application, Kelvin inversion maps the exterior of \(E\), with infinity filled in, to a smooth bounded domain. Composition and the smooth Kelvin boundary factor preserve \(H^{3/2}\). Dirichlet regularity at \(p=2\) gives an \(H^2\) solution, and the normal trace gives \eqref{pub:eq-exterior-trace} after transforming back \cite[Theorem~21.20(i)]{Guermond2015}. These regularity inputs concern fixed smooth domains, not the variable \(C^{1,\gamma}\) boundary.

Let \(p^+\) be the decaying exterior extension of \(p|_{\partial E}\), and set \(k_p=\partial_np-\partial_np^+\). Gluing the interior and exterior fields gives the distributional source \(k_p\delta_{\partial E}\) for the negative Laplacian. Uniqueness of the decaying potential implies \(S_Ek_p=p\) in \(E\). The jump convention is \(\partial_n^-S_Ek-\partial_n^+S_Ek=k\), so
\[
k-k_p=\partial_n^-(u-p)-\mathcal D_E^+\bigl((u-p)|_{\partial E}\bigr).
\]
Equations~\eqref{pub:eq-quadratic-recovery}--\eqref{pub:eq-exterior-trace} imply
\(\|k-k_p\|_{X_E}\le C\|\mathcal L_Ek\|_2\). For general \(X_E\) densities the identities follow by smooth approximation using \eqref{pub:eq-forward} and the trace estimates.

A zero response implies that \(D^2S_Ek\) is constant in connected \(E\), so \(S_Ek\) is harmonic quadratic. Conversely every harmonic polynomial of degree at most two is obtained by the preceding construction. The correspondence is injective: zero interior potential gives zero exterior Dirichlet solution and zero density jump. Thus the kernel has dimension \(1+3+5=9\).

Actual neighboring ellipsoids have equations \((\mathbf x-\mathbf c)^T(\mathbf Q+\mathbf B)(\mathbf x-\mathbf c)<1\). Differentiating their boundaries gives \eqref{pub:eq-tangent-space}. Multiplying a zero combination by \(2|\mathbf Q\mathbf x|\), its odd part forces \(\mathbf b=0\), and its even homogeneous quadratic part forces \(\mathbf B=0\). These nine directions are independent and belong to the kernel because each ellipsoid has a quadratic interior potential. They therefore span the entire kernel, consistently with the infinitesimal ellipsoid characterization \cite{Markenscoff1998a,Markenscoff1998b}. The estimate above is now \eqref{pub:eq-inverse}.
\end{proof}
The projection \(P\) is also bounded on \(Z_{E,\gamma}\). Its basis is smooth, and an \(X_E\) pairing against a basis function becomes an \(L^2(S^2)\) pairing against its smooth image under the multiplier \(\ell+1\). Hence its coefficients are bounded by the strong norm.

\subsection{Exact ellipsoid fitting and distance comparison}
\label{pub:fitting}
Substitute \(\mathbf x+g_a(\mathbf x)\mathbf n_E(\mathbf x)\) into the ellipsoid equation above, with parameter \(a=(\mathbf c,\mathbf B)\). At \(a=0\), the height derivative is \(2|\mathbf Q\mathbf x|>0\) uniformly. The parameter-dependent implicit function theorem gives a unique small height \(g_a\), smooth in the parameters and boundary coordinate. Its image is open and closed in the connected ellipsoid boundary, and hence covers the entire boundary. The phase identification in Lemma~\ref{pub:lem-extension} identifies the enclosed domain with the actual ellipsoid.

The derivative of \(a\mapsto Pg_a\) is the isomorphism \eqref{pub:eq-tangent-space}. Reparametrize by \(v=Pg_a\), restrict to a sufficiently small closed parameter ball, and write
\begin{equation}
\begin{gathered}
\mathcal M_E=\{g(v)=v+\psi(v):v\in\mathcal K_E,\ \|v\|_{X_E}\le r_E\},\\
P\psi(v)=0,\qquad \psi(0)=D\psi(0)=0.
\end{gathered}
\label{pub:eq-chart}
\end{equation}
The fixed-radius chart is compact in every finite regularity class and satisfies \(\|D\psi(v)\|_{\mathcal L(\mathcal K_E,X_E)}\le C\|v\|_{X_E}\). Define
\begin{equation}
v_h=Ph,\qquad g_h=g(v_h),\qquad k_h=h-g_h,
\qquad Pk_h=0.
\label{pub:eq-canonical-fit}
\end{equation}
Boundedness of \(P\) and the chart gives \(\|g_h\|_{Z_{E,\gamma}}+|a(v_h)|\le C\delta\).

If \(v_*\) minimizes \eqref{pub:eq-rho}, comparison with \(g(0)=0\) gives \(\|v_*\|_{X_E}\le2\|h\|_{X_E}\); thus every minimizer is interior for small \(h\). Set \(a_0=v_h-v_*\) and \(b_0=k_h+\psi(v_h)-\psi(v_*)\). Orthogonality yields \(\rho_E(h)^2=\|a_0\|_{X_E}^2+\|b_0\|_{X_E}^2\). Since the Lipschitz constant of \(\psi\) on the joining segment is \(O(\delta)\),
\begin{equation}
\rho_E(h)\le\|k_h\|_{X_E}
\le\sqrt{1+C\delta^2}\,\rho_E(h).
\label{pub:eq-fit-comparison}
\end{equation}
The canonical fit need not minimize the distance; \eqref{pub:eq-fit-comparison} is the comparison with all minimizing fits required by the nonlinear theorem.

\subsection{Phase-matched ellipsoid background}
\label{pub:background}
The relative remainder uses a Lipschitz background estimate only for the finite-dimensional family of actual ellipsoids. Put \(\boldsymbol\Gamma=(\mathbf Q+\mathbf B)^{-1}\), \(\mathbf z=\mathbf y-\mathbf c\), and \(D_\Gamma(t)=\sqrt{\det(\boldsymbol\Gamma+t\mathbf I)}\). The classical formula is \cite[Theorem~2.1]{DiFratta2016}
\begin{equation}
N_{E_a}(\mathbf y)=\frac{\sqrt{\det\boldsymbol\Gamma}}4
\int_{\tau(\mathbf y)}^\infty
\frac{1-\mathbf z^T(\boldsymbol\Gamma+t\mathbf I)^{-1}\mathbf z}{D_\Gamma(t)}\,dt,
\label{pub:eq-classical-ellipsoid}
\end{equation}
where \(\tau=0\) inside and \(\mathbf z^T(\boldsymbol\Gamma+\tau\mathbf I)^{-1}\mathbf z=1\) outside. The integrand vanishes at the exterior lower limit. With \(\mathbf w=(\boldsymbol\Gamma+\tau\mathbf I)^{-1}\mathbf z\) and \(\mathbf n_\tau=\mathbf w/|\mathbf w|\), differentiation gives
\begin{equation}
\begin{aligned}
\mathbf H_a^-&=-\frac{\sqrt{\det\boldsymbol\Gamma}}2
\int_0^\infty\frac{(\boldsymbol\Gamma+t\mathbf I)^{-1}}{D_\Gamma(t)}\,dt,\\
\mathbf H_a^+(\mathbf y)&=-\frac{\sqrt{\det\boldsymbol\Gamma}}2
\int_\tau^\infty\frac{(\boldsymbol\Gamma+t\mathbf I)^{-1}}{D_\Gamma(t)}\,dt
+\frac{\sqrt{\det\boldsymbol\Gamma}}{D_\Gamma(\tau)}
\mathbf n_\tau\otimes\mathbf n_\tau.
\end{aligned}
\label{pub:eq-ellipsoid-hessians}
\end{equation}
At the interface, \(\mathbf H_a^+-\mathbf H_a^-=\mathbf n\otimes\mathbf n\). Matrix formulas remain regular at repeated axes. On bounded closed phase regions, the implicit equation for \(\tau\) has derivative \(-|\mathbf w|^2\) bounded away from zero, so the one-sided Hessians depend smoothly on the parameters.

On the fixed reference phases,
\(D\mathbf U_{g_a}=(\mathbf H_a^\pm\circ\boldsymbol\Phi_{g_a})D\boldsymbol\Phi_{g_a}\).
This composition matches phases rather than comparing opposite sides of a moving jump. To control the unbounded exterior, write \(E_a=\mathbf c+\mathbf L_aB\). Outside a fixed ball containing all nearby ellipsoids,
\[
D^2N_{E_a}(\mathbf y)=\det\mathbf L_a\int_B
D^2\frac1{4\pi|\mathbf y-\mathbf c-\mathbf L_a\boldsymbol\zeta|}\,d\boldsymbol\zeta.
\]
Differentiation of this nonsingular integral gives \(|\mathbf H_a|+|\partial_a\mathbf H_a|\le C|\mathbf y|^{-3}\) and \(|\nabla\mathbf H_a|\le C|\mathbf y|^{-4}\). The coordinate maps and their parameter derivatives are bounded deviations from the identity. Combining the compact-region and tail estimates, then integrating the parameter derivative along \(ta\), proves
\begin{equation}
\|D\mathbf U_{g_a}-D\mathbf U_0\|_{L^\infty(\mathbb R^3)}\le C|a|.
\label{pub:eq-background-bound}
\end{equation}
The norm is an essential norm on the matched phases. No \(L^\infty\) continuity claim for arbitrary moving-boundary Hessians is used.

\subsection{Relative nonlinear estimate and completion}
\label{pub:nonlinear}
Take \(g=g_h\), \(k=k_h\), \(\mathbf Z=\mathbf D_h-\mathbf D_g=\mathbf D_k\), and \(\mathbf R=\mathbf U_h-\mathbf U_g-\mathbf Y_k\). The lift and fitting estimates give
\(\max(\|\mathbf D_h\|_\infty,\|\mathbf D_g\|_\infty)\le C\delta\) and \(\|\mathbf Z\|_2\le C\|k\|_{X_E}\). Subtraction of the two equations \eqref{pub:eq-fixed-domain} and the linearized equation gives
\begin{equation}
\begin{gathered}
\int(D\mathbf R\mathbf A_h):D\boldsymbol\varphi
=\int\sum_{j=1}^4\mathbf S_j:D\boldsymbol\varphi,\\
\begin{aligned}
\mathbf S_1&=-\chi_E[\mathbf Q_h-\mathbf Q_g-\mathbf Q_1(\mathbf Z)],\\
\mathbf S_2&=-D\mathbf U_0[\mathbf A_h-\mathbf A_g-\mathbf A_1(\mathbf Z)],\\
\mathbf S_3&=-(D\mathbf U_g-D\mathbf U_0)(\mathbf A_h-\mathbf A_g),\\
\mathbf S_4&=-D\mathbf Y_k(\mathbf A_h-\mathbf I).
\end{aligned}
\end{gathered}
\label{pub:eq-four-sources}
\end{equation}
The first two sources use \eqref{pub:eq-coefficient-remainders}, the third uses \eqref{pub:eq-background-bound} and \(|a(v_h)|\le C\delta\), and the fourth uses \eqref{pub:eq-Y-bound}. Each has \(L^2\) norm at most \(C\delta\|k\|_{X_E}\). Coercivity and testing with \(\mathbf R\) imply
\begin{equation}
\|D\mathbf R\|_2\le C\delta\|k\|_{X_E}.
\label{pub:eq-R-bound}
\end{equation}

Let \(\mathbf B_h=\mathbf F_h^{-1}\) and \(\mathbf H_h=(D^2N_{\Omega_h})\circ\boldsymbol\Phi_h=D\mathbf U_h\mathbf B_h\), and similarly for \(g\). On \(E\), the exact identity
\begin{equation}
\begin{aligned}
\mathbf H_h-\mathbf H_g-(D\mathbf Y_k-\mathbf H_E\mathbf Z)
={}&D\mathbf R\mathbf B_h+D\mathbf Y_k(\mathbf B_h-\mathbf I)\\
&+(D\mathbf U_g-\mathbf H_E)(\mathbf B_h-\mathbf B_g)\\
&+\mathbf H_E(\mathbf B_h-\mathbf B_g+\mathbf Z)
\end{aligned}
\label{pub:eq-hessian-remainder}
\end{equation}
preserves matrix order. Equations~\eqref{pub:eq-coefficient-remainders}, \eqref{pub:eq-Y-bound}, \eqref{pub:eq-background-bound}, and \eqref{pub:eq-R-bound} control its four terms by \(C\delta\|k\|_{X_E}\). Using \eqref{pub:eq-identify-derivative}, and the fact that \(\mathbf H_g\) is constant on the registered interior, gives
\begin{equation}
\begin{gathered}
\mathcal F_E(h)=V_E^{-1/2}\Pi_E\mathbf H_h,\\
\|\mathcal F_E(h)-\mathcal L_Ek_h\|_{L^2(E)}
\le C\delta\|k_h\|_{X_E}
\le C'\delta\rho_E(h).
\end{gathered}
\label{pub:eq-relative-remainder}
\end{equation}
This is main-text Eq.~(7). Retaining \(\|k_h\|_{X_E}\), rather than an absolute \(O(\delta^2)\) error, is essential when the inclusion lies much closer to another ellipsoid than to the original reference.

Since \(Pk_h=0\), \eqref{pub:eq-gains} and \eqref{pub:eq-relative-remainder} bound \(\|\mathcal F_E(h)\|_2\) between \((c_E-C\delta)\|k_h\|_{X_E}\) and \((C_E+C\delta)\|k_h\|_{X_E}\). Physical averaging differs from reference averaging by \(J_h\). If \(j_-\le J_h\le j_+\), minimization over constant tensors gives
\[
j_-\min_{\mathbf M}\int_E|\mathbf H_h-\mathbf M|_F^2
\le\min_{\mathbf M}\int_E|\mathbf H_h-\mathbf M|_F^2J_h
\le j_+\min_{\mathbf M}\int_E|\mathbf H_h-\mathbf M|_F^2.
\]
Consequently,
\begin{equation}
\sqrt{\frac{V_Ej_-}{|\Omega_h|}}\|\mathcal F_E(h)\|_2
\le q(\Omega_h)\le
\sqrt{\frac{V_Ej_+}{|\Omega_h|}}\|\mathcal F_E(h)\|_2.
\label{pub:eq-physical-average}
\end{equation}
Both factors are \(1+O(\delta)\) by \eqref{pub:eq-J-bounds} and \(|\Omega_h|=\int_EJ_h\). Combining \eqref{pub:eq-fit-comparison}, \eqref{pub:eq-relative-remainder}, and \eqref{pub:eq-physical-average} proves \eqref{pub:eq-local-theorem}. A single positive \(\delta_0\) is chosen to satisfy the finitely many embedding, invertibility, chart, and coefficient restrictions and \(A_{E,\gamma}\delta_0<c_E/2\).

For optimality, fix smooth \(k\perp_{X_E}\mathcal K_E\) and take \(h=tk\). Then the canonical fit is \(g(0)=0\), and
\begin{equation}
\lim_{t\to0}\frac{\rho_E(tk)}{|t|}=\|k\|_{X_E},
\qquad
\lim_{t\to0}\frac{q(\Omega_{tk})}{|t|}=\|\mathcal L_Ek\|_2.
\label{pub:eq-directional-limits}
\end{equation}
Smooth transverse directions are dense in \(\mathcal K_E^\perp\). Choosing unit directions approaching the infimum or supremum in \eqref{pub:eq-gains}, then taking \(t\to0\), proves optimality of both limiting gains. No extremizer is assumed for a general ellipsoid. This completes the proof of Theorem~\ref{pub:thm-local}.

\section{Morphological scaling and spherical spectrum}
\label{pub:scale-spectrum}

\subsection{Tangent-plane field and wavelength scaling}
\label{pub:plane}
For the locally flat reference interface \(z=0\), a real height perturbation \(h(x)=a\cos(\kappa x)\) has first-order interior potential
\begin{equation}
N^{(1)}(x,z)=\frac{a}{2\kappa}e^{\kappa z}\cos(\kappa x),
\qquad z<0.
\label{pub:eq-plane-potential}
\end{equation}
Its nonzero Hessian components are
\[
\begin{aligned}
\partial_{xx}N^{(1)}&=-\frac{a\kappa}{2}e^{\kappa z}\cos(\kappa x),&
\partial_{zz}N^{(1)}&=\frac{a\kappa}{2}e^{\kappa z}\cos(\kappa x),\\
\partial_{xz}N^{(1)}&=-\frac{a\kappa}{2}e^{\kappa z}\sin(\kappa x).
\end{aligned}
\]
Counting both mixed entries in the Frobenius norm gives
\begin{equation}
|\boldsymbol\varepsilon^{(1)}|_F^2
=\frac{\alpha^2\beta^2a^2\kappa^2}{2}e^{2\kappa z},
\qquad
\int_{-\infty}^0|\boldsymbol\varepsilon^{(1)}|_F^2\,dz
=\frac{\alpha^2\beta^2a^2\kappa}{4}.
\label{pub:eq-plane-energy}
\end{equation}
This proves main-text Eq.~(9). The Figure~2A fields are \(m\exp(-2\pi m d)\), with \(m=\kappa/\kappa_0\in\{1,3\}\), \(d=-z/\lambda_0\), \(\lambda_0=2\pi/\kappa_0\), and \(\varepsilon_0=|\alpha|\beta a\kappa_0/\sqrt2\). They are evaluated on the reference half-space, not on a finite-amplitude wavy domain.

For the unitary tangential Fourier transform and \(\boldsymbol\xi\ne0\),
\(\widehat{N^{(1)}}=\widehat h(\boldsymbol\xi)e^{|\boldsymbol\xi|z}/(2|\boldsymbol\xi|)\).
The tangential, mixed, and normal Hessian entries contribute in the ratio \(1:2:1\), giving
\begin{equation}
\int_{-\infty}^0|\widehat{D^2N^{(1)}}(\boldsymbol\xi,z)|_F^2\,dz
=\tfrac12|\boldsymbol\xi|\,|\widehat h(\boldsymbol\xi)|^2.
\label{pub:eq-fourier-half-order}
\end{equation}
Thus the integrated first-order field has the homogeneous half-order boundary weight. The real-cosine and complex-Fourier conventions in \eqref{pub:eq-plane-energy} and \eqref{pub:eq-fourier-half-order} are distinct and consistent.

\subsection{Optimal Sobolev order at every fixed ellipsoid}
\label{pub:sharpness}
For any real \(s\), define on smooth heights
\begin{equation}
\|h\|_{X_{E,s}}^2=\frac1{4\pi}\sum_{\ell\ge0}(\ell+1)^{2s}
\|(h\circ\mathbf T)_\ell\|_{L^2(S^2)}^2,
\qquad
\rho_{E,s}(h)=\min_{g\in\mathcal M_E}\|h-g\|_{X_{E,s}}.
\label{pub:eq-general-order}
\end{equation}
\begin{proposition}[Optimal norm order and linear power]
\label{pub:prop-sharpness}
Within arbitrarily small \(Z_{E,\gamma}\) neighborhoods, no frequency-independent upper bound \(q\le C\rho_{E,s}\) holds when \(s<1/2\), and no positive frequency-independent lower bound \(q\ge c\rho_{E,s}\) holds when \(s>1/2\). The power one in the two-sided estimate cannot be replaced by a stronger power on either side.
\end{proposition}
\begin{proof}
Take \(L^2\)-normalized real spherical harmonics \(Y_\ell\) with \(\|Y_\ell\|_2^2=4\pi\), set \(v_\ell(\mathbf T\boldsymbol\theta)=Y_\ell(\boldsymbol\theta)\), and project \(k_\ell=(I-P)v_\ell\). The pulled-back basis functions of \(\mathcal K_E\) are smooth. Repeated integration by parts with the sphere Laplacian shows that their degree-\(\ell\) coefficients decay faster than any power. The multipliers defining the projections in any fixed \(X_{E,s}\) do not change this conclusion. Therefore
\begin{equation}
\|k_\ell\|_{X_E}\sim(\ell+1)^{1/2},
\qquad
\operatorname{dist}_{X_{E,s}}(k_\ell,\mathcal K_E)
\sim\|k_\ell\|_{X_{E,s}}\sim(\ell+1)^s.
\label{pub:eq-high-degree}
\end{equation}
Here \(\sim\) means ratio tending to one. Proposition~\ref{pub:prop-linear} consequently gives \(\|\mathcal L_Ek_\ell\|_2\asymp(\ell+1)^{1/2}\). Sphere coordinates are used only to prescribe the perturbations; the physical linear operator remains that of the actual ellipsoid.

For each fixed smooth \(k\), the smooth finite-dimensional chart satisfies
\begin{equation}
\lim_{t\to0}\frac{\rho_{E,s}(tk)}{|t|}
=\operatorname{dist}_{X_{E,s}}(k,\mathcal K_E).
\label{pub:eq-distance-derivative}
\end{equation}
The upper limit follows by comparison with \(g(tw)=tw+O(t^2)\). For the lower limit, a minimizing \(v_t\) satisfies \(\|g(v_t)\|_{X_{E,s}}\le2|t|\|k\|_{X_{E,s}}\). Compactness and injectivity of the chart imply \(v_t\to0\); its injective derivative implies \(v_t=O(t)\). Extracting a subsequence of \(v_t/t\) and applying the chart expansion gives the lower limit in \eqref{pub:eq-distance-derivative}.

For each \(\ell\), choose \(t_\ell>0\) small enough that \(\|t_\ell k_\ell\|_{Z_{E,\gamma}}\le1/\ell\) and the field and distance directional limits hold with relative errors tending to zero. Equations~\eqref{pub:eq-directional-limits}, \eqref{pub:eq-high-degree}, and \eqref{pub:eq-distance-derivative} then yield
\begin{equation}
\frac{q(\Omega_{t_\ell k_\ell})}{\rho_{E,s}(t_\ell k_\ell)}
\asymp\ell^{1/2-s}.
\label{pub:eq-order-obstruction}
\end{equation}
This proves both exclusions in the stated admissible class. Along one fixed nonzero transverse direction, both limits in \eqref{pub:eq-directional-limits} are positive. Hence \(q\le C\rho_E^p\) fails for \(p>1\), and \(q\ge c\rho_E^p\) fails for \(p<1\), as \(t\to0\).
\end{proof}

\subsection{Exact sphere coefficients and the three geometric measures}
\label{pub:sphere}
For \(E=B\), normal graphs are radial and \(\mathcal K_B\) is the sum of degrees \(0,1,2\). Write \(f=\sum_\ell f_\ell\) and \(A_\ell^2=(4\pi)^{-1}\|f_\ell\|_2^2\). The single-layer expansion is
\begin{equation}
S_Bf(r\boldsymbol\theta)=\sum_{\ell\ge0}\frac{r^\ell}{2\ell+1}f_\ell(\boldsymbol\theta),
\qquad r<1.
\label{pub:eq-sphere-layer}
\end{equation}
For \(p_\ell=r^\ell f_\ell\), apply Green's identity to its harmonic Cartesian derivatives. Their homogeneity is \(\ell-1\), while the radial and tangential parts of \(\nabla p_\ell\) on the sphere contribute \(\ell^2\|f_\ell\|_2^2\) and \(\ell(\ell+1)\|f_\ell\|_2^2\). Thus
\begin{equation}
\int_B|D^2p_\ell|_F^2
=(\ell-1)\int_{S^2}|\nabla p_\ell|^2
=\ell(\ell-1)(2\ell+1)\|f_\ell\|_2^2.
\label{pub:eq-harmonic-energy}
\end{equation}
Degrees zero and one have zero Hessian; degree two has constant Hessian, removed by \(\Pi_B\). For \(\ell\ge3\), each Hessian entry is a positive-degree harmonic polynomial and has zero ball mean. Cross-degree terms vanish by harmonic orthogonality, and polarization of \eqref{pub:eq-harmonic-energy} eliminates cross terms between orthogonal modes of the same degree. Since \(|B|=4\pi/3\),
\begin{equation}
q_{\rm lin}^2=\|\mathcal L_Bf\|_2^2
=\sum_{\ell\ge3}\omega_\ell A_\ell^2,
\qquad \omega_\ell=\frac{3\ell(\ell-1)}{2\ell+1}.
\label{pub:eq-sphere-spectrum}
\end{equation}
Because \(\omega_\ell\le\tfrac32(\ell+1)\), the identity extends from finite sums to all \(X_B\) traces. For each fixed smooth direction, \(q(\Omega_{tf})^2=t^2\sum_{\ell\ge3}\omega_\ell A_\ell^2+o(t^2)\).

For a pure degree, the denominators and ratios in Figure~2B are
\begin{equation}
\begin{array}{c|c|c}
\text{measure}&\text{value}&q_{\rm lin}/\text{measure}\\ \hline
\text{height RMS}& A_\ell&\sqrt{\omega_\ell}\\
H^{1/2}&\sqrt{\ell+1}\,A_\ell&\sqrt{\omega_\ell/(\ell+1)}\\
\text{slope RMS}&\sqrt{\ell(\ell+1)}\,A_\ell&\sqrt{\omega_\ell/[\ell(\ell+1)]}
\end{array}
\label{pub:eq-mode-normalizations}
\end{equation}
These are three normalizations of the same response. For mixtures, the squared gain for any positive weights \(w_\ell\) is the weighted average
\[
\frac{\sum_\ell\omega_\ell A_\ell^2}{\sum_\ell w_\ell A_\ell^2}
=\sum_\ell\frac{w_\ell A_\ell^2}{\sum_jw_jA_j^2}\frac{\omega_\ell}{w_\ell}.
\]
The weights \(\omega_\ell\) increase strictly. The half-order ratios also increase, since
\[
\frac{\omega_{\ell+1}}{\ell+2}-\frac{\omega_\ell}{\ell+1}
=\frac{3\ell(5\ell+7)}{(\ell+1)(\ell+2)(2\ell+1)(2\ell+3)}>0.
\]
Their value at \(\ell=3\) is \(9/14\), and their limit is \(3/2\). Consequently,
\begin{equation}
\frac3{\sqrt{14}}\rho_{\rm lin}\le q_{\rm lin}\le\sqrt{\frac32}\rho_{\rm lin},
\qquad \rho_{\rm lin}^2=\sum_{\ell\ge3}(\ell+1)A_\ell^2.
\label{pub:eq-sphere-gains}
\end{equation}
The lower gain is attained by degree three. The upper gain is a supremum approached as \(\ell\to\infty\), not attained by a nonzero finite degree. The slope-normalized weights decrease for \(\ell\ge3\): their successive difference is
\(-6(\ell^2-\ell-3)/[(\ell+1)(\ell+2)(2\ell+1)(2\ell+3)]<0\).
At fixed height, the degree-12 to degree-3 ratio is \(\sqrt{\omega_{12}/\omega_3}=2.4819\ldots\); at fixed slope, degree three maximizes the bulk response.

Equations~\eqref{pub:eq-sphere-spectrum} and \eqref{pub:eq-sphere-gains} establish main-text Eqs.~(13)--(15). For finite perturbations the controlled correction is that of Theorem~\ref{pub:thm-local}. In particular, deleting degree-two radial terms is only tangent-space subtraction; an exact finite ellipsoid is fitted by \eqref{pub:eq-chart}.

\section{Finite-shape bounds and convergence}
\label{pub:finite}

\subsection{Residual-field identity and two-sided estimates}
\label{pub:finite-bounds}
Let \(\Omega\subset F\subset\mathbb R^3\), with \(0<V=|\Omega|\le W=|F|<\infty\), and let \(F\) be an ellipsoid. No small-perturbation assumption is imposed. Define
\begin{equation}
t=V/W,\qquad g=\chi_\Omega-t\chi_F,\qquad
\mathbf G=\operatorname{dev}D^2N_g,\qquad
\mathcal B=\frac{2V}{3}(1-V/W).
\label{pub:eq-residual}
\end{equation}
For measurable \(S\subset\Omega\), \(|S|=v_S>0\), and \(T\subset\Omega^c\), set
\begin{equation}
\begin{gathered}
A_S=\int_S|\mathbf G|_F^2\,dV,\quad \mathbf J_S=\int_S\mathbf G\,dV,
\quad E_T=\int_T|\mathbf G|_F^2\,dV,\\
v_R=V-v_S,\qquad R_E=\mathcal B-A_S-E_T.
\end{gathered}
\label{pub:eq-partial-integrals}
\end{equation}
\begin{theorem}[Finite-shape enclosure]
\label{pub:thm-finite}
The residual field satisfies the exact decomposition
\begin{equation}
\mathcal B=\|\mathbf G\|_{L^2(\mathbb R^3)}^2
=Vq(\Omega)^2+
\frac{\left|\int_\Omega\mathbf G\,dV\right|_F^2}{V}
+\int_{\Omega^c}|\mathbf G|_F^2\,dV.
\label{pub:eq-budget}
\end{equation}
The following bounds hold for every \(S,T\) as above:
\begin{equation}
A_S-\frac{|\mathbf J_S|_F^2}{v_S}
\le Vq(\Omega)^2
\le\mathcal B-E_T-\frac1V
\left(|\mathbf J_S|_F-\sqrt{v_RR_E}\right)_+^2,
\label{pub:eq-two-bounds}
\end{equation}
where \(x_+=\max(0,x)\). The unassigned budget \(R_E\) is nonnegative.
\end{theorem}
\begin{proof}
Inside \(F\), \(D^2N_F\) is constant. Hence \(D^2N_g=D^2N_\Omega-tD^2N_F\) on \(\Omega\), and subtraction of the ellipsoidal contribution leaves its variance unchanged. Moreover \(\operatorname{tr}D^2N_g=-(1-t)\) on \(\Omega\), so mean removal acts only on the deviatoric field:
\begin{equation}
Vq^2=\int_\Omega|\mathbf G|_F^2\,dV-
\frac{\left|\int_\Omega\mathbf G\,dV\right|_F^2}{V}.
\label{pub:eq-residual-variance}
\end{equation}
For the unitary Fourier transform,
\(\widehat{D^2N_g}=-(\boldsymbol\xi\otimes\boldsymbol\xi)|\boldsymbol\xi|^{-2}\widehat g\).
The rank-one projector \(\mathbf P_\xi=\boldsymbol\xi\otimes\boldsymbol\xi/|\boldsymbol\xi|^2\) satisfies \(|\mathbf P_\xi-\mathbf I/3|_F^2=2/3\). Plancherel therefore gives
\begin{equation}
\|\mathbf G\|_2^2=\tfrac23\|g\|_2^2,
\qquad
\|g\|_2^2=(1-t)^2V+t^2(W-V)=V(1-V/W).
\label{pub:eq-plancherel-budget}
\end{equation}
Splitting the whole-space norm into the interior and exterior and using \eqref{pub:eq-residual-variance} proves \eqref{pub:eq-budget}.

For the lower estimate, minimization over constants and restriction to \(S\) give
\[
Vq^2=\min_{\mathbf M}\int_\Omega|\mathbf G-\mathbf M|_F^2\,dV
\ge\min_{\mathbf M}\int_S|\mathbf G-\mathbf M|_F^2\,dV
=A_S-|\mathbf J_S|_F^2/v_S.
\]
For the upper estimate, the unassigned budget includes \(\int_{\Omega\setminus S}|\mathbf G|_F^2\). Cauchy--Schwarz on that remainder gives
\[
\left|\int_{\Omega\setminus S}\mathbf G\,dV\right|_F\le\sqrt{v_RR_E},
\qquad
\left|\int_\Omega\mathbf G\,dV\right|_F
\ge\left(|\mathbf J_S|_F-\sqrt{v_RR_E}\right)_+.
\]
Substitute this and \(\int_{\Omega^c}|\mathbf G|_F^2\ge E_T\) into \eqref{pub:eq-budget}. This proves \eqref{pub:eq-two-bounds} and main-text Eqs.~(18)--(21).
\end{proof}
The formulas require only \(\mathbf G\in L^2\), furnished by \eqref{pub:eq-plancherel-budget}. Smoothness and star-shapedness are requirements of the numerical construction in Section~\ref{pub:numerics}, not of this theorem.

\subsection{Convergence under domain exhaustion}
\label{pub:finite-convergence}
Suppose \(S_N\uparrow\Omega\) and \(T_N\uparrow\Omega^c\) up to sets of measure zero, and evaluate the integrals exactly. Let \(\ell_N,u_N\) be the nonnegative square roots of the left and right endpoint expressions in \eqref{pub:eq-two-bounds} after each is divided by \(V\). Because \(\mathbf G\in L^2\),
\[
A_{S_N}\to\int_\Omega|\mathbf G|_F^2,
\quad \mathbf J_{S_N}\to\int_\Omega\mathbf G,
\quad |S_N|\to V,
\quad E_{T_N}\to\int_{\Omega^c}|\mathbf G|_F^2.
\]
Thus \(\ell_N\to q\). For \(S\subset S'\),
\(\min_{\mathbf M}\int_{S'}|\mathbf G-\mathbf M|_F^2\ge\min_{\mathbf M}\int_S|\mathbf G-\mathbf M|_F^2\), so the exact lower functional is nondecreasing. Also \(v_R\to0\) and \(R_E\to0\), whence the mean lower estimate converges to \(|\int_\Omega\mathbf G|_F\) and \(u_N\to q\). The cumulative endpoints
\begin{equation}
L_N=\max_{0\le k\le N}\ell_k,
\qquad U_N=\min_{0\le k\le N}u_k,
\qquad L_N\uparrow q,\quad U_N\downarrow q
\label{pub:eq-monotone-envelopes}
\end{equation}
remain valid and give main-text Eq.~(22). This is a continuum convergence result; it does not prescribe an algebraic rate or bound numerical integration error.

\subsection{Complete far-field accounting on a fixed sphere}
\label{pub:far-field}
Let \(\phi=N_g\) and choose \(B_{R_*}\supset\overline F\). Outside that sphere, \(\Delta\phi=0\) and \(\mathbf G=D^2\phi\). Each \(w_i=\partial_i\phi\) is harmonic. Green's identity on \(B_R\setminus\overline{B}_{R_*}\), followed by \(R\to\infty\), gives
\[
\int_{|\mathbf x|>R_*}|\nabla w_i|^2\,dV
=-\int_{|\mathbf x|=R_*}w_i\partial_n w_i\,dS.
\]
The normal \(\mathbf n\) points out of the sphere, opposite to the exterior-domain normal on its inner boundary. Compact support and \(\int g=0\) give \(\phi=O(|\mathbf x|^{-2})\), \(\nabla\phi=O(|\mathbf x|^{-3})\), and \(D^2\phi=O(|\mathbf x|^{-4})\); the boundary term at infinity vanishes. Summation yields
\begin{equation}
E_{\rm far}(R_*)=\int_{|\mathbf x|>R_*}|\mathbf G|_F^2\,dV
=-\int_{|\mathbf x|=R_*}\nabla\phi\cdot(D^2\phi\,\mathbf n)\,dS.
\label{pub:eq-far-field}
\end{equation}
For an exterior exclusion region \(\Omega\subset\Omega_N^+\subset B_{R_*}\),
\begin{equation}
E_{T_N}=\int_{B_{R_*}\setminus\Omega_N^+}|\mathbf G|_F^2\,dV+E_{\rm far}(R_*),
\qquad T_N=\mathbb R^3\setminus\Omega_N^+.
\label{pub:eq-exterior-split}
\end{equation}
The fixed surface contribution can be reused across levels. This is an exact integration identity, not a truncated physical problem or an imposed boundary condition.

\section{Numerical evaluation and independent FEM comparison}
\label{pub:numerics}

\subsection{Finite-geometry potential evaluation}
\label{pub:potential-numerics}
For a target strictly separated from a smooth interface, integration by parts in the source coordinate gives
\begin{equation}
\partial_jN_\Omega(\mathbf x)=-\frac1{4\pi}\int_{\partial\Omega}
\frac{n_j(\mathbf y)}{|\mathbf x-\mathbf y|}\,dS_{\mathbf y},
\qquad
H_{ij}(\mathbf x)=\frac1{4\pi}\int_{\partial\Omega}
\frac{(x_i-y_i)n_j(\mathbf y)}{|\mathbf x-\mathbf y|^3}\,dS_{\mathbf y}.
\label{pub:eq-surface-field}
\end{equation}
This evaluates the actual finite geometry, not its first variation. The integrand is nonsingular at a separated target but becomes concentrated near the boundary. Symmetry and the phase traces \(\operatorname{tr}\mathbf H=-1\) inside and \(0\) outside are numerical diagnostics. The one-sided jump is \(\mathbf H^+-\mathbf H^-=\mathbf n\otimes\mathbf n\): tangential derivatives of the continuous gradient agree, and the trace jump is one. With \eqref{pub:eq-elastic-reduction}, the traction jump is \(\alpha[3K-\beta(\lambda+2\mu)]\mathbf n=0\).

For a surface of revolution \(\mathbf y=(a(\vartheta)\cos\varphi,a(\vartheta)\sin\varphi,z_s(\vartheta))\), the outward vector area is
\[
\mathbf n\,dS=(-az_s'\cos\varphi,-az_s'\sin\varphi,aa')\,d\vartheta\,d\varphi.
\]
At \(\mathbf x=(p,0,z)\), define \(A=p^2+a^2+(z-z_s)^2\), \(B=2pa\), and
\(I_j=\int_0^{2\pi}\cos^j\varphi\,(A-B\cos\varphi)^{-3/2}d\varphi\), \(j=0,1,2\).
Writing \(d=(p-a)^2+(z-z_s)^2\), \(s=(p+a)^2+(z-z_s)^2\), and \(m=1-d/s\), the ring identities are
\begin{equation}
\begin{gathered}
J=\frac{4K(m)}{\sqrt s},\qquad T=4\sqrt s\,E(m),\qquad
I_0=\frac{4E(m)}{d\sqrt s},\\
I_1=\frac{AI_0-J}{B},\qquad
I_2=\frac{A^2I_0-2AJ+T}{B^2}.
\end{gathered}
\label{pub:eq-ring}
\end{equation}
Here \(K(m),E(m)\) are complete elliptic integrals, not the elastic bulk modulus or reference ellipsoid. For small \(B/A\), cancellation in \(I_1,I_2\) is avoided using
\[
I_j=2\pi A^{-3/2}\sum_{n\ge0}\frac{(3/2)_n}{n!}(B/A)^n c_{n+j},
\qquad c_{2k}=\binom{2k}{k}4^{-k},\quad c_{2k+1}=0.
\]
The meridional integrands for \((H_{rr},H_{\varphi\varphi},H_{zz},H_{rz},H_{zr})\), before division by \(4\pi\), are respectively
\[
\begin{gathered}
(-az_s')(pI_1-aI_2),\quad az_s'a(I_0-I_2),\quad (z-z_s)aa'I_0,\\
aa'(pI_0-aI_1),\quad(z-z_s)(-az_s')I_1.
\end{gathered}
\]
These identities specify the common azimuthal reduction. Meridional and target quadratures require separate refinement, particularly near interfaces.

For an axisymmetric tensor field, the Cartesian volume mean is
\begin{equation}
\langle\mathbf H\rangle_\Omega=\operatorname{diag}(M,M,Z),
\quad M=\left\langle\tfrac12(H_{rr}+H_{\varphi\varphi})\right\rangle_\Omega,
\quad Z=\langle H_{zz}\rangle_\Omega,
\label{pub:eq-H-mean}
\end{equation}
and
\begin{equation}
q^2=\left\langle(H_{rr}-M)^2+(H_{\varphi\varphi}-M)^2+(H_{zz}-Z)^2+2H_{rz}^2\right\rangle_\Omega.
\label{pub:eq-axis-q}
\end{equation}
The two transverse component means are not subtracted separately. For a positive radial shape \(r(\boldsymbol\omega)\), volume quadrature uses \(\mathbf x=s r(\boldsymbol\omega)\boldsymbol\omega\) and \(dV=r^3s^2\,ds\,dS_{\boldsymbol\omega}\). Whole-interior response calculations cover \(0\le s\le1\); the smaller cores in the enclosure calculation have the separate role specified in Section~\ref{pub:dyadic}.

\subsection{Supporting spherical and finite-amplitude checks}
\label{pub:mode-checks}
The finite-amplitude calculations use
\begin{equation}
\begin{gathered}
r(\theta)=s_h\left[1+\sum_\ell\sqrt{2\ell+1}\,A_\ell P_\ell(\cos\theta)\right],\\
s_h=\left[\tfrac12\int_{-1}^1(1+h(\mu))^3\,d\mu\right]^{-1/3}.
\end{gathered}
\label{pub:eq-finite-radial}
\end{equation}
This fixes the volume at \(4\pi/3\); for a zero-mean direction, \(s_h=1+O(\|h\|^2)\), leaving the leading spectrum unchanged. The source integral used the transformation \(\vartheta=\vartheta_0+w\sinh t\), with \(w\) proportional to target separation, to concentrate Gaussian nodes near the kernel. Its small-\(B/A\) ring evaluation retained 16 series terms for \(B/A<0.06\). These settings belong to the supporting finite-amplitude calculation, not to the dyadic enclosure levels below.

Actual ellipsoid fits used the positive radial intersection with \((\mathbf x-\mathbf c)^T\mathbf Q_F(\mathbf x-\mathbf c)=1\):
\[
r_F(\boldsymbol\omega)=\frac{b+\sqrt{b^2-ad}}a,
\quad a=\boldsymbol\omega^T\mathbf Q_F\boldsymbol\omega,
\quad b=\boldsymbol\omega^T\mathbf Q_F\mathbf c,
\quad d=\mathbf c^T\mathbf Q_F\mathbf c-1.
\]
Fits were compared at harmonic cutoffs 36 and 60 and 140 and 200 polar nodes, with selected full nine-parameter fits. This numerical fitting is distinct from the exact local minimum in \eqref{pub:eq-rho}.

The finite-amplitude calculations include exact sphere and spheroid controls, four combined source/volume quadrature levels, separate full-azimuth surface cubature at 48 targets, and amplitude sequences for degrees three, twelve, and their mixture. The largest reported componentwise ring/cubature difference is \(1.92\times10^{-13}\), including targets with relative radial gaps down to \(10^{-4}\). The spheroid mean-Hessian discrepancy is \(7.83\times10^{-15}\). These compare quadrature implementations of \eqref{pub:eq-surface-field}, not independent governing-equation solvers.

As a representative finite-amplitude check, at \(A_3=0.004\) the numerically evaluated gain is
\begin{equation}
\frac{q_{\rm num}}{\rho_{B,\rm num}}
=0.80169815\ldots,
\qquad \frac3{\sqrt{14}}=0.80178373\ldots .
\label{pub:eq-finite-gain-check}
\end{equation}
The finite-geometry value is slightly below the linearized lower endpoint, as permitted by the nonlinear correction. The amplitude sequences approach the spectral prediction, with different signs and orders of correction for odd, even, and mixed directions. They do not evaluate a universal nonlinear constant or certify an admissible finite amplitude.

For a compact quadrature comparison, Table~\ref{pub:tab-mode-quad} shows the source/volume refinement for the two mixed directions \((A_3,A_{12})=(0.004,0.002)\) and \((0.004,0.0008)\). All values use the normalization \eqref{pub:eq-q}. Agreement of successive values indicates numerical stability.
\begin{table}[!htbp]
\centering\small
\caption{Finite-amplitude quadrature comparison. The orders are bulk polar, bulk radial, and source meridional, respectively. Both responses use the normalization \eqref{pub:eq-q}.}
\label{pub:tab-mode-quad}
\begin{tabular}{ccc}
\toprule
Orders & \(q\), \((0.004,0.002)\) & \(q\), \((0.004,0.0008)\)\\
\midrule
32/16/120 & 0.0102192712 & 0.0071598867\\
64/32/200 & 0.0102192712 & 0.0071598867\\
112/56/300 & 0.0102192712 & 0.0071598867\\
160/80/420 & 0.0102192712 & 0.0071598867\\
\bottomrule
\end{tabular}
\end{table}

\subsection{Application geometries and containing ellipsoids}
\label{pub:geometries}
The presentation order is B (protrusion), C (indentation), and D (necked particle), with A the reference spheroid. Let \(R_0=(3V_A/4\pi)^{1/3}\); all coordinates and geometric constants in this subsection are divided by this fixed reference radius. Set
\begin{equation}
\begin{gathered}
a=(3/2)^{2/3}=1.3103706971\ldots,\qquad
b=(3/2)^{-1/3}=0.8735804647\ldots,\\
z_A(r)=a\sqrt{1-r^2/b^2}.
\end{gathered}
\label{pub:eq-reference-geometry}
\end{equation}
Thus \(ab^2=1\) and \(V_A/R_0^3=4\pi/3\). The upper profiles of B and C are
\begin{equation}
\begin{gathered}
z_{B,C}(r)=z_A(r)\pm h\psi(r/w),\qquad h=0.35,\qquad w=0.55b,\\
\psi(t)=\begin{cases}\exp(1-1/(1-t^2)),&0\le t<1,\\0,&t\ge1.\end{cases}
\end{gathered}
\label{pub:eq-feature}
\end{equation}
The opposite surface is unchanged. These are axial height changes over a common footprint, so
\begin{equation}
\frac{V_B-V_A}{R_0^3}=\frac{V_A-V_C}{R_0^3}
=2\pi h w^2\int_0^1t\psi(t)\,dt=0.1024606055\ldots .
\label{pub:eq-paired-volume}
\end{equation}
Consequently \(V_B/R_0^3=4.2912508103\ldots\), \(V_C/R_0^3=4.0863295993\ldots\), and the relative changes are \(\pm2.4461\%\), both relative to \(V_A\).

For D, the connected two-lobe profile is
\begin{equation}
\begin{gathered}
r^2(z)=\sqrt{c^4+4d^2z^2}-d^2-z^2,\qquad
\frac dc=\sqrt{\frac{1+\sqrt{1-(3/5)^2}}2},\\
c=1.3663236441\ldots,\qquad d=1.2962084209\ldots .
\end{gathered}
\label{pub:eq-necked-geometry}
\end{equation}
The scale fixes its volume at \(4\pi/3\). The neck radius is \(r_0=0.4320694736\ldots\), the maximum lobe radius is \(r_1=0.7201157894\ldots\), and \(r_0/r_1=0.60\).

The fixed containing ellipsoids used by the enclosure calculation are
\begin{equation}
\begin{aligned}
F_B:&\quad r^2/0.887^2+(z-0.127)^2/1.535^2\le1,\\
F_C:&\quad r^2/b^2+z^2/a^2\le1,\\
F_D:&\quad r^2/0.8453^2+z^2/2.284^2\le1.
\end{aligned}
\label{pub:eq-containing-ellipsoids}
\end{equation}
They are selected from geometry without FEM input. The containment calculation gives a maximum level bound \(0.9979125674<1\) for B and, after the substitution \(s=z^2\), a quadratic-inequality minimum bounded below by \(0.00262035037>0\) for D. C is contained exactly in the parent spheroid. The parent outlines in Figure~3A specify the paired feature construction; they are not the computational containing ellipsoids in general.

\subsection{Dyadic regions and fixed-region accuracy checks}
\label{pub:dyadic}
These three geometries are star-shaped about the stated origin. For \(N=0,\ldots,4\), use
\begin{equation}
\tau_N=0.1\,2^{-N},\qquad S_N=(1-\tau_N)^{1/3}\Omega,
\quad\Omega_N^+=(1+\tau_N)^{1/3}\Omega,
\quad T_N=\mathbb R^3\setminus\Omega_N^+.
\label{pub:eq-dyadic}
\end{equation}
Both omitted interface layers have volume \(\tau_NV\). The origin, \(\Omega\), and \(F\) remain fixed. A common \(R_*=4R_0\) contains all three \(F\) and \(\Omega_0^+\); \eqref{pub:eq-exterior-split} includes the complete far field, with the fixed surface term reused at all levels. The potential integral is evaluated by the surface/ring representation in Section~\ref{pub:potential-numerics}.

Quadrature is compared at fixed regions. If \(\delta_{\rm quad,N}\) is the larger absolute change of the two squared endpoints between the selected resolutions, the acceptance diagnostic is
\begin{equation}
\frac{\delta_{\rm quad,N}}{u_N^2-\ell_N^2}\le0.05.
\label{pub:eq-quadrature-criterion}
\end{equation}
The maximum recorded ratio is \(5.0646\times10^{-4}\), or \(0.0506\%\) of the squared width; the finest level receives an additional check. At \(\tau=0.025\), shifting the split from \(4R_0\) to \(6R_0\) changes the total exterior integral by at most \(1.12\times10^{-16}\). The shell integral agrees with the difference of the far-field surface terms. The unit-sphere test gives \(0.04363323129985823\), against \(0.04363323129985824\) analytically. Maximum recorded interior/exterior trace defects are \(1.01\times10^{-13}\) and \(8.30\times10^{-14}\). These are recorded floating-point consistency checks, not directed error bounds; the selected-pair diagnostic does not certify either integral.

\subsection{Working intervals and contraction}
\label{pub:intervals}
Table~\ref{pub:tab-intervals} gives the raw accepted endpoints. Lower values are displayed downward and upper values upward to six decimal places. Figure~3C uses the unrounded endpoints evaluated at each refinement level.

\begin{table}[!htbp]
\centering\small
\caption{Working numerical intervals for the common omitted-volume sequence.}
\label{pub:tab-intervals}
\setlength{\tabcolsep}{4pt}
\begin{tabular}{rccc}
\toprule
\(\tau\) (\%) & B: protrusion & C: indentation & D: necked\\
\midrule
10 & [0.045575, 0.208631] & [0.063440, 0.085858] & [0.145941, 0.264029]\\
5 & [0.048941, 0.165658] & [0.069641, 0.082011] & [0.150360, 0.226000]\\
2.5 & [0.050667, 0.129324] & [0.072894, 0.079620] & [0.152524, 0.197672]\\
1.25 & [0.051545, 0.101765] & [0.074564, 0.078144] & [0.153595, 0.178937]\\
0.625 & [0.051989, 0.082406] & [0.075409, 0.077286] & [0.154127, 0.167759]\\
\bottomrule
\end{tabular}
\end{table}

All 15 intervals contain their FEM reference. The first-to-final reductions of \(u_N-\ell_N\) are \(81.3\%\), \(91.6\%\), and \(88.5\%\), in B--C--D order. For the squared-field budget, define
\begin{equation}
W_N^{(2)}=u_N^2-\ell_N^2,
\qquad p_N=\log_2\!\left(W_{N-1}^{(2)}/W_N^{(2)}\right),\quad N\ge1.
\label{pub:eq-observed-contraction}
\end{equation}
The final-step orders are \(0.913,0.932,0.942\). Figure~\ref{pub:fig-contraction} shows the observed contraction, not a proved universal first-order rate. The \(L^2\) exhaustion proof in Section~\ref{pub:finite-convergence} supplies convergence without a rate; an algebraic rate additionally requires quantitative control of field concentration in the omitted layers.
\begin{figure}[!htb]
\centering
\includegraphics[width=0.92\linewidth,keepaspectratio]{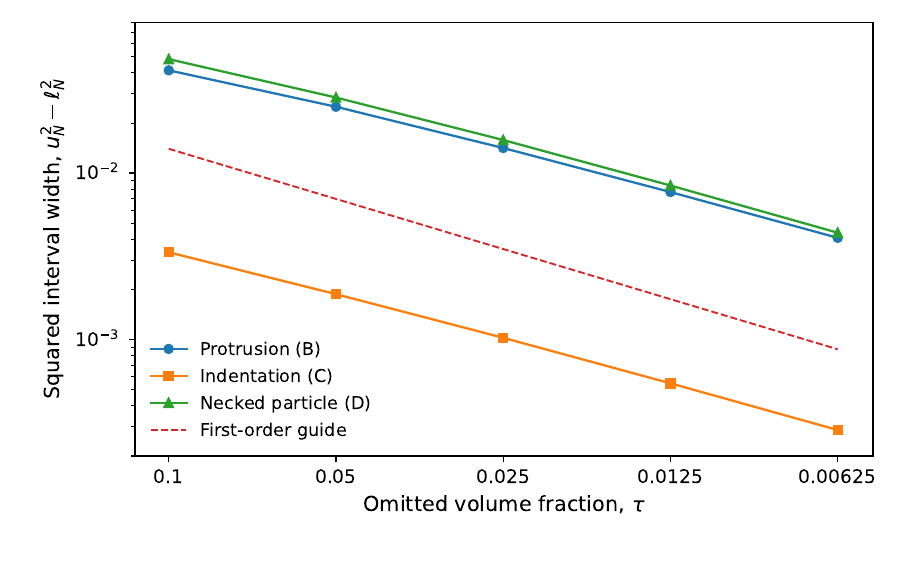}
\caption{Contraction of the squared interval width under the dyadic omitted-volume rule. Curves use the unrounded endpoints for protrusion, indentation, and necked particle. The first-order reference is not a fit. Quadrature sensitivity is evaluated separately at fixed regions.}
\label{pub:fig-contraction}
\end{figure}

\subsection{Independent axisymmetric finite-element formulation}
\label{pub:fem}
The FEM solves displacement equilibrium directly, without the potential operator or interval endpoints. The FEM parameters are \(\lambda=\mu=1\), \(\alpha=10^{-3}\), and \(\beta=5/3\). The full positive-\(r\), positive/negative-\(z\) meridian is retained:
\[
\begin{gathered}
D_L=\{(r,z):r\ge0,\ r^2+z^2<L^2\},\\
u_r=u_z=0\text{ on }r^2+z^2=L^2,\qquad
u_r=0\text{ on }r=0.
\end{gathered}
\]
Here \(u_z\) remains free on the rotation axis. The strain components and weak equation are
\begin{equation}
\begin{gathered}
\varepsilon_{rr}=\partial_ru_r,\qquad \varepsilon_{zz}=\partial_zu_z,
\qquad\varepsilon_{\theta\theta}=u_r/r,
\qquad\varepsilon_{rz}=\tfrac12(\partial_zu_r+\partial_ru_z),\\
\int_{D_L}\mathbb C\bigl(\boldsymbol\varepsilon(\mathbf u)-\alpha\mathbf I\chi_\Omega\bigr):
\boldsymbol\varepsilon(\mathbf v)\,2\pi r\,dr\,dz=0.
\end{gathered}
\label{pub:eq-fem-weak}
\end{equation}
Six-node quadratic isoparametric triangles represent geometry and displacement. Curved elements fit the analytic inclusion meridian and outer sphere. The exterior mesh grows geometrically. Domain studies preserve the mesh through \(3R_0\). The mesh sequence uses 32/64/128 angular subdivisions and 6/12/24 interior radial layers; C adds a 256/48 level because of its rim field. Element assembly checks sampled Jacobian positivity.

Statistics use quadrature gradients, not smoothed nodal fields. The Cartesian mean is \(\operatorname{diag}(m_t,m_t,m_z)\), where \(m_t=\langle(\varepsilon_{rr}+\varepsilon_{\theta\theta})/2\rangle_\Omega\) and \(m_z=\langle\varepsilon_{zz}\rangle_\Omega\). Then
\begin{equation}
q_{\rm FEM}^2=\frac{\left\langle(\varepsilon_{rr}-m_t)^2+(\varepsilon_{\theta\theta}-m_t)^2+
(\varepsilon_{zz}-m_z)^2+2\varepsilon_{rz}^2\right\rangle_\Omega}{\alpha^2\beta^2}.
\label{pub:eq-fem-observable}
\end{equation}
Figure~3B shows the element RMS of the corresponding local normalized fluctuation.

\subsection{FEM benchmarks, domain and mesh checks}
\label{pub:fem-checks}
The outer radius is varied through \(L/R_0=8,16,24\). For a unit spherical inclusion in a concentric clamped sphere, the exact finite-domain interior strain is
\[
\boldsymbol\varepsilon_{\rm in}=\frac{\alpha\beta}{3}(1-L^{-3})\mathbf I
\]
in radius-normalized coordinates. At \(L=24R_0\), the sphere and reference-spheroid full-interior tensor errors relative to their infinite-domain solutions are \(0.00910\%\) and \(0.00834\%\). The sphere error relative to the finite-clamped solution is \(0.00481\%\). Their residual heterogeneities, \(2.76\times10^{-5}\) and \(2.10\times10^{-5}\), provide numerical-floor controls.

On the common 128/24 mesh, increasing \(L\) from \(16R_0\) to \(24R_0\) changes \(q\) by \(1.04\times10^{-5}\%\), \(1.16\times10^{-6}\%\), and \(4.86\times10^{-6}\%\) for B, C, and D. The selected results are given in Table~\ref{pub:tab-fem}.
\begin{table}[!htbp]
\centering\small
\caption{Independent FEM responses and observed final-mesh sensitivity.}
\label{pub:tab-fem}
\begin{tabular}{lrrr}
\toprule
Case & \(q_{\rm FEM}\) & mesh change & DOFs\\
\midrule
B: protrusion & 0.0524349864 & 0.0151\% & 52,174\\
C: indentation & 0.0762543627 & 0.0595\% & 161,598\\
D: necked & 0.1546568936 & 0.0025\% & 52,174\\
\bottomrule
\end{tabular}
\end{table}
The C value uses the 256/48 mesh. A 4-by-4 to 6-by-6 Gaussian-product quadrature check on C changes \(q\) by approximately \(6.7\times10^{-10}\) relatively. Additional checks verify linear scaling under doubling/reversing \(\alpha\), recovery of stored variances from quadrature fields, and stiffness-work closure at approximately \(10^{-13}\) or better. Exact potential solutions enter only after the displacement solve for benchmarking. The same independent FEM responses and field maps are used throughout the dyadic refinement. These sensitivities are not rigorous continuum FEM error intervals.

\section{Research-process documentation and reproducibility}
\label{pub:process}

\subsection{Scope of AI review and partial formalization}
\label{pub:review-scope}
Main-text Section~3.2 distinguishes the researchers' scientific decisions from AI-assisted formulation, derivation, computation, and review. The staged AI review examined a local stability estimate for smooth perturbations of a fixed ellipsoid. The AI reviewer first received the definitions, operator, registration, and target estimate, and reconstructed an argument before examining the candidate derivation. The subsequent AI review identified no substantive gap within that scope. This provided an additional analytical check; separate AI review contexts do not ensure independent model errors. The \(C^{1,\gamma}\) two-sided theorem is established by the written proof in Sections~\ref{pub:setting}--\ref{pub:scale-spectrum}.

Partial Lean formalization established specified mathematical components rather than the full local theorem. Table~\ref{pub:tab-formal-scope} summarizes the checked statements and the aspects of the analytical argument that they do not establish.
\begin{table}[!htbp]
\centering\small
\caption{Scope of selected formalization results.}
\label{pub:tab-formal-scope}
\begin{tabular}{@{}>{\raggedright\arraybackslash}p{0.23\linewidth}>{\raggedright\arraybackslash}p{0.36\linewidth}>{\raggedright\arraybackslash}p{0.34\linewidth}@{}}
\toprule
Component & Checked statement & Not established by that component\\
\midrule
Sobolev realization & A continuous linear reconstruction takes the completion of smooth gradients to an \(L^6\) vector field with the specified weak gradient and Sobolev bound. & Identification with the actual registered moving-domain potential; the transverse inverse estimate.\\
Ellipsoid tangents & Differentiation of the quadratic level constraint, injectivity of the normal-velocity parametrization, and a nine-dimensional algebraic range. & Boundedness in \(X_E\), identification with the full response kernel, or the nonlinear stability theorem.\\
\bottomrule
\end{tabular}
\end{table}

\subsection{Separation of analytical and computational evidence}
\label{pub:evidence}
The local and finite-shape theorems are analytical statements with the assumptions given in Sections~\ref{pub:setting}--\ref{pub:finite}. The tangent-plane and sphere calculations establish their respective first-order formulas. The supporting finite-amplitude calculations use the same potential representation on actual geometries. Their departure from the sharp linearized lower gain, \eqref{pub:eq-finite-gain-check}, distinguishes that gain from a finite-amplitude bound. Comparing ring and full-azimuth cubatures tests the numerical evaluation of the shared potential representation.

For the dyadic enclosure study, the theoretical endpoints were evaluated without using FEM results. Displacement-based FEM provides an independent comparison by solving the equilibrium equations for the same prescribed geometries.

\EndArxivSupplement
\end{document}